\documentclass[reqno,10pt]{amsart}
\usepackage{amsmath,amssymb}
\usepackage[mathscr]{euscript}

\usepackage{epic}

\theoremstyle{plain}
\newtheorem{theorem}{Theorem}[section]
\newtheorem{proposition}[theorem]{Proposition}
\newtheorem{lemma}[theorem]{Lemma}
\newtheorem{corollary}[theorem]{Corollary}
\numberwithin{equation}{section}
\numberwithin{theorem}{section}

\newcommand{\mc}[1]{{\mathcal #1}}

\newcommand{\ms}[1]{{\mathscr #1}}
\newcommand{\bb}[1]{{\mathbb #1}}
\newcommand{\mf}[1]{\mathfrak{#1}}

\DeclareMathOperator{\tr}{Tr}
\newcommand{\varsh}{\mathop{\rm sh}\nolimits}
\newcommand{\varch}{\mathop{\rm ch}\nolimits}

\newcommand{\ad}{\mathop{\rm ad}\nolimits}
\newcommand{\Ad}{\mathop{\rm Ad}\nolimits}

\newcommand{\Dom}{\mathop{\rm Dom}\nolimits}
\newcommand{\End}{\mathop{\rm End}\nolimits}
\newcommand{\gap}{\mathop{\rm gap}\nolimits}
\newcommand{\rank}{\mathop{\rm rank}\nolimits}
\newcommand{\Span}{\mathop{\rm Span}\nolimits}
\newcommand{\Ent}{\mathop{\rm Ent}\nolimits}

\newcommand{\id}{{1 \mskip -5mu {\rm I}}}
\newcommand{\ind}{\mathbf{ 1}}
\newcommand{\ide}{\textrm{id}}

\renewcommand{\epsilon}{\varepsilon}
\renewcommand{\tilde}{\widetilde}

\begin{document}

\title[Ergodicity for QMS]{Trace distance ergodicity for \\ quantum Markov semigroups}

\author [L.\ Bertini]{Lorenzo Bertini}
\address{Lorenzo Bertini 
\hfill\break \indent
Dipartimento di Matematica, Universit\`a di Roma La Sapienza,
\hfill\break \indent
I-00185 Roma, Italy}
\email{bertini@mat.uniroma1.it}
\author [A.\ De Sole]{Alberto De Sole}
\address{Alberto De Sole
\hfill\break \indent
Dipartimento di Matematica \& INFN, Universit\`a di Roma La Sapienza,
\hfill\break \indent
I-00185 Roma, Italy}
\email{desole@mat.uniroma1.it}
\author [G.\ Posta]{Gustavo Posta}
\address{Gustavo Posta\hfill\break \indent
  Dipartimento di Matematica, Universit\`a di Roma La Sapienza,
  \hfill\break \indent
  I-00185 Roma, Italy}
\email{gustavo.posta@uniroma1.it}

\thanks{The work of A.\ De Sole  has been  supported  by  PRIN 2022S8SSW2}

\noindent


\begin{abstract}
We discuss the quantitative ergodicity of quantum Markov semigroups 
in terms of the trace distance from the stationary state, providing a general criterion based on the spectral decomposition
of the Lindblad generator.
We then apply this criterion 
to the bosonic and fermionic Ornstein-Uhlenbeck semigroups 
and to a family of quantum Markov semigroups parametrized by semisimple Lie algebras
and their irreducible representations,
in which the Lindblad generator is given by the adjoint action of the Casimir element.
\end{abstract}

\maketitle
\thispagestyle{empty}

\section{Introduction}
\label{s:1}

A \emph{quantum Markov semigroup} (QMS) is a one-parameter semigroup $(P_t)_{t\geq0}$
of completely positive operators on a $C^*$-algebra $\mc A$.
QMS are typically used to describe the non-unitary evolution of quantum systems,
as in the case of (weak) interaction with thermal baths 
or in measurement processes \cite{BP,SL}.
See also \cite{Ba} for the relevance of QMS in the context of quantum information theory.
The generator of a QMS is referred to as the Lindblad generator.
We focus to the case in which $P_t$ is self-adjoint with respect to the GNS inner product
induced by a state $\sigma$.
In particular, $\sigma$ is an invariant state for $P_t$.

Given an initial state $\rho$, we are interested in the exponential convergenze
of its evolution $\rho P_t$ to the invariant state $\sigma$,
that we quantify in terms of the \emph{trace distance},
the non-commutative version of the total variation distance.
In the context of classical Markov semigroups, this topic has become a central issue in probability theory that it has been addressed by different methods such as: coupling arguments, functional inequalities, and spectral theory. 
We refer to \cite{Peres} for a review on the subject.

To the best of our knowledge, a decay to equilibrium for QMS
based on a non-commutative version of coupling arguments has not been discussed in literature.

The methods based on functional inequalities can be directly extended to the non-commutative case.
Indeed, in view of the quantum Pinsker inequality, see e.g.\ \cite[Thm.11.9.5]{wilde},
the exponential convergence in trace distance can be deduced
by the exponential decay of the quantum relative entropy of $\rho P_t$
with respect to $\sigma$.
This exponential decay of the quantum relative entropy is in fact equivalent to the validity
of the modified log-Sobolev inequality.
Sufficient criteria for the latter, along the same lines of the celebrated Bakry-\'Emery criterion \cite{BGL}, 
have been discussed in \cite{Br:Ga:Ju,CM1,CM2,DR,Li:Ju:La,Wi:Zh,Coreani}.

The methods based on spectral theory require a detailed knowledge of the spectral decomposition 
of the generator of the semigroup, but, when such knowledge is available, they provide sharp bounds
on the decay rate.
In the present paper we deduce the exponential convergence to equilibrium
in trace distance for QMS by spectral methods.
More precisely, following \cite{TKRWV}, we introduce a non-commutative version of the $\chi^2$-divergence
and observe that it provides,
as in the commutative case, 
an upper bound for the trace distance between states.
We then show that the exponential decay of the $\chi^2$-divergence between $\rho P_t$ and $\sigma$
can be deduced from the spectral decomposition of the Lindblad generator $L$.
We emphasize that the eigenvectors of $L$
are normalized with respect to the GNS inner product,
but the previous estimate requires a control on their operator norm.
This is the non-commutative counterpart of an $L^\infty$-bound 
on the $L^2$-normalized eigenfunctions of the Markov generator discussed in \cite[\S12.6]{Peres}.
We also remark that, as in the commutative case,
the quantum relative entropy of a state $\rho$ with respect to $\sigma$
is controlled by the $\chi^2$-divergence of $\rho$ with respect to $\sigma$,
see e.g. \cite[Thm.8]{TKRWV}.
Hence, this approach also yields the exponential decay of the quantum relative entropy
of $\rho P_t$ with respect to the stationary state $\sigma$ with a quantitative decay rate.

As discussed in \cite[Thm.3.1]{CM1} and \cite[Thm.3]{Ali}, the explicit construction of a Lindblad generator that is self-adjoint with respect to the GNS inner product induced by a given state $\sigma$ requires the eigenvalues and eigenvectors of the modular operator associated to $\sigma$.
Identifying the state $\sigma$ with the corresponding density matrix,  we recall that the modular operator associated to $\sigma$ is the map $a\mapsto \sigma a\sigma^{-1}$.
For this reason, concrete examples of QMS that are self-adjoint with respect to given GNS inner products
have been so far limited to special models.
%
If $h$ is the Hamiltonian of a harmonic oscillator and $\sigma$
is the thermal state $e^{-\beta h}/\tr(e^{-\beta h})$ for some $\beta>0$, 
the bosonic Ornstein-Uhlenbeck semigroup is a QMS that is self-adjoint with respect to the GNS inner product
induced by $\sigma$.
Its construction has been carried out, by  Dirichlet form techniques, in \cite{CFL}.
The corresponding decay rate to equilibrium in entropic sense
has been proven in \cite[Theorem 8.5]{CM1}
using Mehler type formulae and intertwining relationships.
As we show in Section \ref{s:qou}, the decay to equilibrium in trace distance,
with the same rate as in \cite{CM1},
can also be obtained
from the spectral decomposition of the Lindblad generator.

The other paradigmatic example discussed in \cite{CM1}
is the fermionic Ornstein-Uhlenbeck semigroup,
which is self-adjoint with respect to the GNS inner product associated to a free fermion thermal state.
In this case the underlying Hilbert space is finite dimensional
and the exponential convergence in trace distance holds uniformly
with respect to the initial state.
The sharp exponential rate has been computed in \cite[Thm.8.6]{CM1},
again by intertwining relationships. 
In Section \ref{s:fou} we show that, also in the fermionic case, 
the decay to equilibrium in trace distance can be obtained
from the spectral decomposition of the Lindblad generator,
with the same rate as in \cite{CM1}.

A simple and popular QMS is the so-called depolarizing channel \cite{Chu}.
The exponential decay of the quantum relative entropy has been obtained in \cite{MS}.
It has also been analyzed in \cite{CM2,CM3,DR} in terms of a lower bound
on a ``Ricci curvature'', however the decay rate of the quantum relative entropy 
obtained there is not sharp.
Section \ref{s:la} contains the main novelty of the present paper:
the analysis of the sharp decay rate for a family of QMS parametrized by a finite dimensional semisimple Lie algebra $\mf g$ and a finite dimensional irreducible $\mf g$-module $V$,
in which the Lindblad generator is given by the adjoint action of the Casimir element.
As discussed in \cite{Ba}, these QMS can also be obtained by `transference' from heat kernel of the connected compact Lie group associate to the Lie algebra $\mf g$.
Such QMS are self-adjoint with respect to the normalized Hilbert-Schmidt inner product 
on $\End(V)$.
The corresponding generators can be viewed as
non-commutative versions of the Laplace-Beltrami operator on suitable manifolds.
The depolarizing channel mentioned above
is obtained as the special case when $\mf g=\mf{sl}_2$ and $V$ is its defining
2-dimensional representation.
It appears that the exponential decay of the quantum relative entropy for such family of QMS
cannot be obtained by applying the criteria in \cite{CM1,CM2,DR}.
On the other hand, the representation theory of finite dimensional semisimple Lie algebras
gives detailed information on the spectrum of the corresponding Lindblad generators.
In fact, sharp bounds on the eigenvalues and eigenvectors of the Lindblad generator
lead to exponential decay of the quantum $\chi^2$-divergence.
As a consequence, we deduce the exponential convergence of $\rho P_t$ to $\sigma$ 
both in trace distance and in entropic sense for all the QMS in the family.
We emphasize that for this class of QMS
the bound that we obtain for the rate of the exponential convergence in trace distance 
coincides with the inverse of the spectral gap of the Lindblad generator, and it is therefore  optimal.
More precisely, we show that there are reals $g_0,A>0$,
depending on the Lie algebra $\mf g$ but independent of the $\mf g$-module $V$,
such that, for any $t\geq0$ and any initial state $\rho$,
\begin{equation*}
\mathrm{d}_\mathrm{tr}\big(\rho P_t,\sigma\big)
  \le A \, e^{-g_0t}
  \,,
\end{equation*}
where $\mathrm{d}_\mathrm{tr}$ denotes the trace distance and the stationary state $\sigma$ 
is the normalized trace on $V$.

\section{Convergence in trace distance for quantum
  Markov semigroups} 
\label{s:2}

In this section we recall some basic facts about QMS and the trace distance between states.
We then prove a general inequality yielding the exponential ergodicity of QMS in trace distance, generalizing the bound stated in \cite[\S12.6]{Peres} for classical Markov semigroups.  
The discussion will be carried out
in the context of QMS on the $C^*$-algebra of linear operators
on a finite dimensional Hilbert space.
This setting, that avoids the functional analytic technicalities of infinite dimensional $C^*$-algebras,
is sufficient to accomodate the Fermi Ornstein-Uhlenbeck semigroup that we discuss in Section \ref{s:fou}, and the Lie algebra based QMS that we introduce in Section \ref{s:la}.
As in the case of Markov chains with finite state space,
the main emphasis is to obtain sharp bounds relative to 
the dimension of the underlying Hilbert space.
In contrast, the Bose Ornstein-Uhlenbeck semigroup acts on an infinite dimensional $C^*$-algebra and will be discussed in Section~\ref{s:qou}, where we provide the needed additional details.

We refer to \cite{CM1} for the more general setting of abstract finite dimensional $C^*$-algebras
and for a discussion on the self-adjointness of QMS with respect to other inner products than 
the GNS inner product considered here.
See also \cite{gustavo,Fa} and references therein
for a discussion of QMS on infinite dimensional $C^*$-algebras.

\subsection*{Quantum Markov semigroups and Lindblad generators}

Let $H$ be a finite dimensional Euclidean space over $\bb C$, the
inner product on $H$ is denoted by $\cdot$ and the corresponding
Euclidean norm by $|\cdot|$.
Let $\mc A:= \mathrm{End} (H)$ be the algebra of linear
operators on $H$. The algebra $\mc A$ becomes a finite-dimensional
$C^*$-algebra when it is endowed with the operator norm
$\|a\|:=\sup_{|x|=1} |ax|$ and the involution
$\mc A \ni a \mapsto a^*\in \mc A$, where $a^*$ is the adjoint of $a$,
i.e.\ $a^*x\cdot y = x\cdot a y$. The identity in $\mc A$ is denoted
by $\ind$.
An element $a\in\mc A$ is
\emph{self-adjoint} if $a=a^*$; the collection of self-adjoint
elements in $\mc A$ is denoted by $\mc A_\mathrm{sa}$. An element
$a\in \mc A_\mathrm{sa}$ is \emph{positive} if there exists $b\in\mc A$ such that
$a=bb^*$; equivalently if all the eigenvalues of $a$ are positive.
The cone of positive elements in $\mc A$ is denoted
by $\mc A_+$.
A linear map $\phi \colon \mc A\to \mc A$ is \emph{positive} if
$\phi(\mc A_+) \subset \mc A_+$.  The map $\phi$ is \emph{completely
  positive} if for any positive integer $k$ the map
$\ide \otimes \phi \, \colon \mathrm{End} (\bb C^{k})
\otimes \mc A \to
\mathrm{End} (\bb C^{k}) \otimes \mc A$ is positive.

A \emph{quantum Markov semigroup} (QMS) on $\mc A$ is a one-parameter,
strongly continuous, semigroup $(P_t)_{t\ge 0}$ on $\mc A$ such that
for each $t\ge 0$ the map $P_t\colon \mc A\to \mc A$ is completely
positive and satisfies $P_t \ind =\ind$. By the properties of
completely positive maps on $C^*$-algebras, if $(P_t)_{t\ge 0}$ is a
QMS then $\|P_t a\| \le \|a\|$ for any $a\in\mc
A$ and $t\ge 0$, i.e.\ $(P_t)_{t\ge 0}$ is a contraction semigroup
on $\mc A$.
Moreover, $(P_ta)^*= P_t a^*$ for any $a\in \mc A$ and $t\ge 0$.

According to the Hille-Yoshida theorem, a strongly continuous
contraction semigroup $(P_t)_{t\ge 0}$ on the Banach space $\mc A$ 
has the form $P_t=\exp\{t L\}$ for some \emph{generator} 
$L\in \mathrm{End}(\mc A)$.
If $(P_t)_{t\ge 0}$ is a QMS on $\mc A$, then its generator $L$ is characterized
by the following structure theorem \cite{Lindblad76}.
There exists $h\in \mc A_\mathrm{sa}$, a positive
integer $N$, and $\ell_j \in \mc A$, $j=1,\ldots, N$, such that $L$ has the form
\begin{equation}
  \label{3.2}
  L a = i [h,a] + \sum_j \big(
  [\ell_j^*,a] \,\ell_j + \ell^*_j \,[a,\ell_j]   \big)
\end{equation}
where $[a,b]=ab-ba$.
Conversely, for any choice of $h\in\mc A_\mathrm{sa}$, $N$, and $\ell_j\in \mc A$, $j=1,\ldots,N$,
the operator $L$ on $\mc A$ defined by the right hand side of \eqref{3.2} generates a QMS.
In the physical literature, the operator $L$ in \eqref{3.2} is
called Lindblad generator, $h$ Hamiltonian, and the $\ell_j$ jump operators.
%

\subsection*{Reversible quantum Markov semigroup} 

Let $\mc A'$ be the dual of $\mc A$ and denote by $\mc A_{+}'$ its 
positive cone, i.e.\ the set of elements $\sigma\in \mc A'$
satisfying $\sigma( \mc A_+)\subset \bb R_+$.
A \emph{state} $\sigma$ on $\mc A$
is an element of $\mc A'_+$ satisfying $\sigma(\ind)=1$. 
The set of states on $\mc A$ is denoted by $\mc A'_{+,1}$. 
Letting $\tr$ be the trace on $\mc A$, we here identify
the state $\sigma$ with its \emph{density matrix} in $\mathcal{A}$, still denoted
by $\sigma$,  satisfying $\sigma\geq0$ and $\tr(\sigma)=1$. Namely, 
$\sigma(a) = \tr(\sigma\, a)$.
The state $\sigma$ is \emph{faithful} when
$\sigma(aa^*)=0$ implies $a=0$, equivalently when 
the eigenvalues of the density matrix $\sigma$ are strictly positive.    
If $\sigma$ is a faithful state, the corresponding GNS inner product
$\langle\cdot,\cdot\rangle_\sigma$ on $\mc A$ is defined by 
$\langle a,b\rangle_\sigma :=\sigma(b a^*) =\tr(a^* \sigma b)$.
When equipped with this inner product, $\mc A$
becomes an Euclidean space over $\bb C$ denoted by $L_2(\sigma)$, while
the corresponding Euclidean norm is denoted by $\|\cdot\|_\sigma$.

Let $(P_t)_{t\ge 0}$ be a QMS on
$\mc A$; by duality $(P_t)_{t\ge 0}$ defines a semigroup
on $\mc A'$, i.e.\ $(\rho P_t)(a) := \rho(P_t a)$.
The conditions of complete positivity and the normalization $P_t\ind
=\ind$ imply that $(P_t)_{t\ge 0}$ preserves the set of states. 
If the generator of $P_t$ is the operator $L$ in \eqref{3.2},
then the generator of the dual of $(P_t)_{t\ge 0}$ is 
\begin{equation*}
  L^\dagger \rho = -i [h,\rho]
  + \sum_j \big(
  [\ell_j\rho,\ell_j^*] +  [\ell_j,\rho\ell_j^*]  \big).
\end{equation*}
The state $\sigma$ is \emph{invariant} if $\sigma P_t =\sigma$ for all
$t\ge 0$ or, in terms of the Lindblad generator,
if $L^\dagger\sigma =0$.
If $\sigma$ is an invariant state then the semigroup
$(P_t)_{t\ge 0}$ is also a contraction semigroup on $L_2(\sigma)$, i.e.\
$\|P_t a\|_\sigma \le \|a\|_\sigma$.
This statement is direct consequence of the Kadison-Schwarz
inequality $(P_t a)( P_t a^*) \le P_t (aa^*)$, $t\ge 0$,
$a\in \mc A$, which holds by the complete positivity of $P_t$ and $P_t\ind =\ind$.

The QMS $(P_t)_{t\ge 0}$ on $\mc A$ is \emph{reversibile} with respect to the state $\sigma$ if $P_t$ is
a self-adjoint operator on $L_2(\sigma)$ for $t\geq0$.
Equivalently, the generator $L$
in \eqref{3.2} is self-adjoint as an operator on $L_2(\sigma)$.
In the commutative case, reversibility of a Markov
semigroup is equivalent to the invariance under time reversal of the
law of the associated stationary process. The following statement,
proven by a direct computation, provides the analogue in the present
non-commutative setting.

\begin{lemma}
  Let $(P_t)_{t\ge 0}$ be a QMS on $\mc A$.
  Given $n\in\bb N$, $0\le t_1\le \cdots \le
  t_n$, and $a_1,\cdots, a_n \in \mc A$, set
  \begin{equation*}
    \begin{split}
      & a_1(t_1) \cdots a_n(t_n) :=
          P_{t_1} \big( a_1 P_{t_2-t_1} \big( a_2 \cdots   
          P_{t_n-t_{n-1}} \big(a_n\big)\cdots \big)\big)
      \\
      & a_1(t_n) \cdots a_n(t_1) := P_{t_1} \big( \cdots
      P_{t_{n-1}-t_{n-2}} \big( P_{t_n-t_{n-1}} \big(a_1\big) a_2\big)
      \cdots a_n\big).
    \end{split}
    \end{equation*}
  Then $(P_t)_{t\ge 0}$ is reversible with respect to $\sigma$ if and
  only if for each $n\in\bb N$, $0\le t_1\le \cdots \le t_n\le T$, and
  $a_1,\cdots, a_n \in \mc A$,
  \begin{equation*}
    \sigma\big( a_1(t_1) \cdots a_n(t_n) \big) =
    \sigma\big(a_1(T-t_1) \cdots a_n(T-t_n) \big).
  \end{equation*}
\end{lemma}

Fix $\beta> 0$ and consider the faithful state represented by the density matrix
$\sigma := \exp\{-\beta h\}/\tr\big( \exp\{-\beta h\}\big)$, where
$h\in \mc A_\mathrm{sa}$ is the Hamiltonian appearing in \eqref{3.2}. 
Then the first term in the right-hand side of \eqref{3.2} 
corresponds to the Heisenberg evolution, that is skew-adjoint as on
operator on $L_2(\sigma)$. We denote by $L_0$ the second term, i.e.
\begin{equation}
  \label{3.2.0}
  L_0 a = \sum_j \big([\ell_j^*,a] \,\ell_j + \ell^*_j \,[a,\ell_j]   \big).
\end{equation}
In the next statement we provide a necessary and sufficient condition for the 
self-adjointness of $L_0$ as an operator on $L_2(\sigma)$. 

\begin{lemma}
  \label{t:p1}
  Let $\mathrm{Ad}_\sigma\colon \mc A \to \mc A$ be the \emph{modular operator}
  $a\mapsto \sigma a\sigma^{-1}$.
The generator $L_0$ in \eqref{3.2.0} is self-adjoint in $L_2(\sigma)$
if and only if
\begin{equation}
  \label{calb2}
 2 \sum_{j} \big(
  \ell_j \otimes \mathrm{Ad}_\sigma(\ell_j^*) 
  - \ell_j^* \otimes \ell_j
  \big)
  =
  \id\otimes
  \sum_{j} \big(\mathrm{Ad}_\sigma(\ell_j^*\ell_j)-\ell^*_j\ell_j
 \big),
\end{equation}
where $\id$ is the identity operator on $\mathcal{A}$.
Furthermore, a sufficient condition for the self-adjointness of $L_0$ in $L_2(\sigma)$
is 
\begin{equation}
  \label{calb}
 \sum_{j} 
  \ell_j \otimes \mathrm{Ad}_\sigma(\ell_j^*) 
  =
  \sum_j \ell_j^* \otimes \ell_j
\,.
\end{equation}
\end{lemma}

\begin{proof}
By direct computation 
$\langle b,L_0 a\rangle_\sigma=\langle L_0b,a\rangle_\sigma$ for all $b\in\mc A$
is equivalent to
$$
\sum_j\big(
\sigma[\ell_j^*,a]\ell_j+\sigma\ell_j^*[a,\ell_j]
\big)
=
\sum_j\big(
[\ell_j\sigma a,\ell_j^*]+[\ell_j,\sigma a\ell_j^*]
\big)
$$
which can be rewritten as
$$
2\sum_j(
\Ad_{\sigma}(\ell_j^*)\sigma a\ell_j
-\ell_j\sigma a\ell_j^*
)
-\sum_j(
\Ad_{\sigma}(\ell_j^*\ell_j)
-\ell_j^*\ell_j
) \sigma a
=
0
$$
The above equation holds for all $a\in\mc A$ if and only if \eqref{calb2} holds.

For the last assertion, it is enough to notice that \eqref{calb}, when applied  two times, 
implies that the right-hand side of \eqref{calb2} vanishes.
\end{proof}


As discussed in \cite[Thm.3.1]{CM1} and \cite[Thm.3]{Ali}, given a faithful state $\sigma$ there is a
simple algorithm to construct self-adjoint Lindblad generators on
$L_2(\sigma)$.
Find eigenvectors of $\Ad_\sigma$, i.e.\ solve the
equation
\begin{equation}
  \label{eAs}
  \sigma v \sigma^{-1} =e^\omega \, v
\end{equation}
where $\omega\in \bb R$ and $v\in \mc A$.
By the self-adjointness of $\sigma$, if $(\omega, v)$ is a
solution to \eqref{eAs} then $( -\omega, v^*)$ is also a solution.
Given a collection $\{v_j\}$ of eigenvectors of $\Ad_\sigma$
that is closed with respect to $*$-adjunction, set
\begin{equation}
\label{faa}
  L_0 a =\sum_{j} e^{\omega_j/2} 
  \big( [v_j^*,a] v_j + v_j^* [a,v_j] \big)
\end{equation}
where $\omega_j$ is the eigenvalue associated to $v_j$.
The generator in \eqref{faa} has the form \eqref{3.2.0} with
$\ell_j = e^{\omega_j/4} v_j$. It is then straightforward to check
that condition \eqref{calb} holds, so that $L_0$ is indeed
self-adjoint on $L_2(\sigma)$.
In \cite[Thm.3.1]{CM1} and \cite[Thm.3]{Ali} it is shown that any Lindblad generator that is
self-adjoint on $L_2(\sigma)$ can be written in the form \eqref{faa}
for a suitable collection of $v_j$'s that are eigenvectors of $\Ad_\sigma$.

Let  $\sigma=\exp\{-\beta h\}/\tr(e^{- \beta h})$ for some
$h\in \mc A_\textrm{sa}$ and $\beta> 0$.
In particular, setting  $\textrm{ad}_h(v) := [h,v]$, we have 
$\Ad_\sigma =\exp\{-\beta \textrm{ad}_h\}$.
Equation \eqref{eAs} thus amounts to find the eigenvectors of
$\textrm{ad}_h$, i.e.\ to solve $\textrm{ad}_h(v) = \gamma v$ and then
set  $\omega=-\beta \gamma$.
The operators $v_j$ can be constructed from the spectral
decomposition of $h$. 
Using the bra-ket Dirac notation, 
let $h=\sum_n\epsilon_n |n\rangle\langle n|$ be the spectral decomposition of $h$.
Then $|m\rangle \langle n|\in \mc A$ is an eigenvector of $\ad_h$
with eigenvalue $\epsilon_m-\epsilon_n$. Given a
Hermitian matrix $(c_{m,n})_{m,n}$, the family
$\big\{c_{m,n} |m\rangle \langle n| \big\}_{m,n}$ is closed under
$*$-adjunction.
The Lindblad generator self-adjoint in $L_2(\sigma)$ associated to this family as in \eqref{faa} then reads 
\begin{equation*}
  L_0 a =\sum_{m,n} e^{\tfrac{\beta}{2} (\epsilon_n-\epsilon_m)} 
  |c_{m,n}|^2 
  \big( 2 \langle m |a|m  \rangle
  \, |n \rangle \langle n|
  - a |n \rangle \langle n| -|n \rangle \langle n|a  \big).
\end{equation*}

We finally mention that if $L_0$ has the form \eqref{faa} and
$\sigma = \exp\{-\beta h\}/\tr\big( e^{-\beta h}\big)$ for some $\beta>0$, then
$L_0$ commutes with the Heisenberg generator $i\ad_h$. Hence the
semigroups generated by $L_0$ and $i\ad_h$ commute \cite{Ali}.

\subsection*{Ergodicity}
We first consider the second term on the decomposition
\eqref{3.2} of the Lindblad generator $L$ that we denote by $L_0$ as
in \eqref{3.2.0}.
We further assume that $L_0$ is reversible with respect to the faithful state $\sigma$ and we regard it as an operator on $L_2(\sigma)$.
The corresponding \emph{Dirichlet form}
$E_\sigma \colon L_2(\sigma) \to \bb R_+$ is defined by
\begin{equation}
  \label{dirf}
  E_\sigma(a) := - \langle L_0 a, a\rangle_\sigma
  = \sum_j \big\langle [\ell_j^*, a] , [\ell_j^*, a] \big\rangle_\sigma 
  = \sum_j  \big\| \, [\ell_j^*, a ] \, \big\|_\sigma^2 
\end{equation}
where $\| \cdot \|_\sigma$ is the norm in $L_2(\sigma)$ and the second
equality follows by a direct computation.

Let $(P_t)_{t\ge 0}$ be the semigroup generated by $L_0$.
By the spectral theorem, the following statements are equivalent.
\begin{itemize}
\item [(i)]
   For each $a\in \mc A$, $\lim_{t\to\infty}\|P_t a-\sigma(a )\ind\|_\sigma=0$;
\item[(ii)] $0$ is a simple eigenvalue of $L_0$;
\item [(iii)] $E_\sigma(a) =0$ implies that $a$ is a multiple of $\ind$. 
\end{itemize}
The semigroup $P_t=\exp\{t L_0\}$, $t\geq0$, is \emph{ergodic} when any (and hence
all) of the above conditions is met.
In view of \eqref{dirf}, the semigroup $(P_t)_{t\geq0}$ is ergodic if
and only if the commutant (or centralizer) in $\mc A$ of the family $\{\ell_j^*\}$ is
$\bb C \ind$.
Loosely speaking, in order to ensure ergodicity the Lindblad generator has to be defined with enough jump operators.

Fix $\gamma>0$. 
Again by the spectral theorem, the following statements are equivalent.
\begin{itemize}
\item [(i)]
  For any $t\ge 0$ and any $a\in \mc A$
\begin{equation}
  \label{2.1}
  \big\| P_t a - \sigma(a) \ind \big\|_\sigma
  \le e^{-\gamma t}
  \big\| a - \sigma (a) \ind \big\|_\sigma;
\end{equation}
\item[(ii)] 
  $\gamma\le \textrm{gap}(-L_0)$, the second smallest eigenvalue of $-L_0$;
\item [(iii)]
  $  \gamma \, \big\| a - \sigma( a) \ind \|_\sigma ^2 \le E_\sigma(a)$
for all $a \in \mc A$.  
\end{itemize}

If $\sigma=\exp\{-\beta h\}/\tr(e^{-\beta h})$ for some $\beta>0$ then, as observed before, $L_\textrm{0}$ and $i \ad_h$ commute.
Hence, if $L_\textrm{0}$ is ergodic
then also the semigroup generated by $L= i \ad_h + L_0$ in
\eqref{3.2} is ergodic in the sense that statement (i) holds.
Moreover, if \eqref{2.1} holds for the semigroup generated by $L_0$,
then the same inequality holds for the semigroup generated by $L$.

\subsection*{Trace distance and quantum $\chi^2$-divergence} 

In the context of QMS, the most relevant
question on the velocity of convergence to the invariant state is the following: 
starting from an arbitrary state $\rho$, 
how large $t$ needs to be in order that $\rho P_t$ is close to the
stationary state $\sigma$? The $L_2$ ergodicity discussed above
does not really answer this question because it involves the 
$L_2$ distance with respect to the stationary state.
We address this issue by considering the trace distance on the set of
states and analyzing its behavior as $t\to \infty$.

The \emph{trace distance} on the set of states $\mc A'_{+,1}$ is defined by
\begin{equation}
 \label{dtrvvar} 
  \mathrm{d}_\mathrm{tr} (\rho,\sigma) :=
  \frac 12
  \sup_{\|a\|=1}  \big| \rho(a) -\sigma(a) \big|
  = \frac 12
  \sup_{\|a\|=1} \big| \tr\big( (\rho-\sigma) \, a \big)  \big|
 \end{equation}
where $\|\cdot\|$ is the operator norm on $\mc A$.
Equivalently, 
\begin{equation*}
  \mathrm{d}_\mathrm{tr} (\rho,\sigma) =\frac 12\,
  \tr \big( \sqrt{ (\rho-\sigma)^2} \, \big) =\frac 12 \sum_i |\lambda_i|
\end{equation*}
where $\lambda_i$ are the eigenvalues of $\rho-\sigma$. 
The normalizing factor $1/2$ in \eqref{dtrvvar} has been chosen so that $\mathrm{d}_\mathrm{tr}$  reduces to the total variation distance in the commutative case.

Fix a faithful state $\sigma$ and recall that $L_2(\sigma)$ denotes
the Euclidean space over $\bb C$ obtained by endowing $\mc A$ with the
GNS inner product  $\langle a,b\rangle_\sigma := \tr(a^* \sigma b)$.
Following \cite{TKRWV}, the \emph{quantum $\chi^2$-divergence} is the map $\mathrm{D}_{\chi^2} (\cdot | \sigma)\colon \mc A'_{+,1}\to \bb R_+$ defined by
\begin{equation}
 \label{dc2var} 
  \mathrm{D}_{\chi^2} (\rho | \sigma) :=
  \sup_{\|a\|_\sigma =1}  \big| \rho(a) -\sigma(a) \big|^2
  = \sup_{\|a\|_{\sigma}=1}
  \big| \langle \sigma^{-1}\rho-\ind, a \rangle_\sigma \big|^2
 \end{equation}
 where we emphasize that $\|\cdot\|_{\sigma}$ is the norm on
 $L_2(\sigma)$.
 By duality,
 \begin{equation*}
\mathrm{D}_{\chi^2} (\rho | \sigma)
= \big\| \sigma^{-1}\rho -\ind \big\|_{\sigma}^2
=\tr\big(
(\rho-\sigma)\sigma^{-1}(\rho-\sigma)
\big)
.
\end{equation*}
Note that $\mathrm{D}_{\chi^2} (\cdot | \cdot)$ is not symmetric so
that it cannot be used to introduce a distance on the set of
states.
However, as in the commutative case, the trace distance can be
bounded in terms of the quantum $\chi^2$-divergence.

\begin{lemma}
  \label{tr<c2}
  Let $\sigma$ be a faithful state. Then for each $\rho\in\mc
  A'_{+,1}$
  \begin{equation*}
    \mathrm{d}_\mathrm{tr} (\rho,\sigma)^2 \le
    \frac 14 \,\mathrm{D}_{\chi^2} (\rho|\sigma).
  \end{equation*}
\end{lemma}
A more general statement is proven in \cite[Lem.5]{TKRWV}.
For the reader's convenience,
we provide here an elementary variational proof.
\begin{proof}
  Since $\sigma\in \mc A'_{+,1}$, for each $a\in\mc A$,
  \begin{equation*}
    \|a\|_{\sigma}^2=
    \langle a,a \rangle_\sigma = \sigma(a a^*)
    \le \|a a^*\| = \|a\|^2.
  \end{equation*}
  The statement thus follows from the variational
  representations \eqref{dtrvvar} and
  \eqref{dc2var}. 
\end{proof}

We observe that the quantum relative entropy 
can be controlled in terms of the $\mathrm{D}_{\chi^2}$-divergence.
We refer to \cite[Prop.6]{TKRWV} for the proof of the following statement.
\begin{lemma}
Let $\sigma$ be a faithful state and
set $\Ent(\rho|\sigma)=\tr(\rho(\log\rho-\log\sigma))$.
Then
$$
\Ent(\rho|\sigma)
\leq
\mathrm{D}_{\chi^2}(\rho|\sigma)
\,.
$$
\end{lemma}

The next statement is the non-commutative version of a classical bound
for reversible Markov chains, see e.g.\ \cite[\S12.6]{Peres}.
Under suitable conditions on the spectral decomposition of the
generator, it can be used to deduce the exponential ergodicity of the
semigroup $(P_t)_{t\ge 0}$ in trace distance \emph{uniformly} with
respect to the initial state.

\begin{theorem}
  \label{t:tr<}
  Let $(P_t)_{t\ge 0}$ be an ergodic QMS on $\mc A$ reversible with respect to the faithful state $\sigma$.
  Let the corresponding generator $L$ have spectral decomposition
  \begin{equation}\label{eq:specdec}
    -L=\sum_{j\geq0}\lambda_j\langle f_j,\cdot\rangle_\sigma f_j
  \end{equation}
  where $\lambda_0=0$ and $f_0=\ind$.
  Then for any $\rho\in \mc A'_{+,1}$ and $t\ge 0$.
\begin{equation}
\label{eq:bound}
   \mathrm{d}_\mathrm{tr} (\rho P_t,\sigma)^2 \le
    \frac 14 \, \sum_{j\ge 1} e^{-2\lambda_j t} \big| \rho(f_j) \big|^2.
\end{equation}
In particular,
\begin{equation}
\label{eq:sup}
  \sup_{\rho\in \mc A'_{+,1} }\,\mathrm{d}_\mathrm{tr} (\rho P_t,\sigma)^2 \le
    \frac 14 \, \sum_{j\ge 1} e^{-2\lambda_j t} \|f_j\|^2.
\end{equation}
\end{theorem}
We emphasize that in the bound \eqref{eq:sup} the set $\{f_j\}_{j\ge 0}$
is an orthonormal basis in $L_2(\sigma)$, so that
$\|f_j\|_{\sigma}=1$ while $\|f_j\|\geq1$ is the operator norm of $f_j$. 
By ergodicity $\lambda_j>0$ for $j\geq1$.

\begin{proof}
  Since $(\rho P_t)(a) = \rho(P_t a)$,
  by definition \eqref{dc2var}, the invariance of $\sigma$,
  the self-adjointness of $P_t$ on
  $L_2(\sigma)$, and the spectral theorem,
  \begin{equation*}
    \begin{split}
      \mathrm{D}_{\chi^2} (\rho P_t|\sigma) & =
      \sup_{\|a\|_{\sigma}=1}
      \big|
      \langle \sigma^{-1}\rho-\ind, P_t a \rangle_\sigma \big|^2
      =
      \sup_{\|a\|_{\sigma}=1}
      \big|\langle P_t(\sigma^{-1}\rho)-\ind, a \rangle_\sigma\big|^2
      \\
      &=
      \big\| P_t (\sigma^{-1}\rho) -\ind
      \big\|_{\sigma}^2=
      \Big\| \sum_{j\ge 0} e^{-\lambda_j t}
      \langle f_j, \sigma^{-1}\rho\rangle_\sigma f_j  - \ind
      \Big\|_{\sigma}^2
      \\
      &= 
      \Big\| \sum_{j\ge 1} e^{-\lambda_j t} \langle f_j ,
      \sigma^{-1}\rho\rangle_\sigma f_j \Big\|_{\sigma}^2 =
    \sum_{j\ge 1} e^{-2\lambda_j t}
    \big| \langle f_j, \sigma^{-1}\rho\rangle_\sigma\big|^2
    \\
    &
    =\sum_{j\ge 1} e^{-2\lambda_j t} \big|\rho(f_j)\big|^2,
    \end{split}
    \end{equation*}
    where we used $\rho(f_j^*)=\overline{\rho(f_j)}$.
    The statement now follows from Lemma~\ref{tr<c2}.
  \end{proof}

\begin{corollary}
Let $(P_t)_{t\geq0}$ be a QMS with Lindblad generator
$L = i\ad_h+L_0$, where $L_0$ is reversible with respect to the state $\sigma=e^{-\beta h}/\tr(e^{-\beta h})$
for some $\beta>0$.
Assume that $-L_0$ has spectral decomposition given by the right-hand side of \eqref{eq:specdec}.
Then the bound \eqref{eq:bound} holds
for any $\rho\in \mc A'_{+,1}$ and $t\ge 0$.
\end{corollary}
\begin{proof}
As observed before, 
under the assumptions of the Corollary,
the semigroups generated by $L_0$ and $i\ad_h$ commute.
Since $\mathrm{d}_\mathrm{tr} (\rho e^{i\ad_h t},\sigma)
=\mathrm{d}_\mathrm{tr} (\rho,\sigma)$, the claim follows directly from Theorem \ref{t:tr<}.
\end{proof}

\section{Fermi Ornstein-Uhlenbeck semigroup}
\label{s:fou}

Given a positive integer $N$, let $a_k,a^*_k$, $k=1,\dots,N$ 
be a family of operators
satisfying the canonical anticommutation relations (CAR)
\begin{equation}\label{car}
  \{a_h,a_k\}=\{a^*_h,a^*_k\}=0, \qquad
  \{a_h,a^*_k\}=\delta_{h,k} \ind,
\end{equation}
where $\{a,b\}=ab+ba$ denotes the anticommutator of $a$ and $b$.
Let $\mc A$ be the $C^*$-algebra generated by the $a_k$, $a^*_k$, $k=1,\dots,N$.
This algebra can be realized as the $C^*$-algebra of operators on the Hilbert space
$H=(\bb C^2)^{\otimes N}$ in terms of the Jordan-Wigner transformation:
$$
a_k=\sigma_z^{\otimes(k-1)}\otimes \sigma^+\otimes\id_2^{\otimes(N-k)}
\,\,,\,\,\,\,
a^*_k=\sigma_z^{\otimes(k-1)}\otimes \sigma^-\otimes\id_2^{\otimes(N-k)}
\,,
$$
where
$$
\sigma^+=\left(\begin{array}{cc} 0&1 \\ 0&0 \end{array}\right)
\,,\,\,
\sigma^-=\left(\begin{array}{cc} 0&0 \\ 1&0 \end{array}\right)
\,,\,\,
\sigma_z=\left(\begin{array}{cc} 1&0 \\ 0&-1 \end{array}\right)
\,,\,\,
\id_2=\left(\begin{array}{cc} 1&0 \\ 0&1 \end{array}\right)
\,.
$$

Let $n_k=a_k^*a_k$, $k=1,\dots,N$, be the fermionic number operators,
which are pairwise commuting, self-adjoint, and satisfy $n_k^2=n_k$.
For a collection $\{\omega_k\}_{k=1}^N$ of reals,
consider the free fermionic Hamiltonian
$h=\sum_k\omega_kn_k$
and set $\sigma=e^{- h}/\tr(e^{- h})$, where we have absorbed the dependence of the inverse temperature $\beta$ in the parameters $\omega_k$.
Following \cite{CM1}, we introduce a QMS reversible with respect to $\sigma$.
Let $w$ be the self-adjoint and unitary element of $\mc A$ given by
$w=\prod_{k=1}^N(2n_k-\ind)$,  and set $v_k=wa_k$ and $v_k^*=a_k^*w$.
Observe that, since $w$ commutes with $h$,
$$
\ad_h(v_k)=-\omega_k v_k
\,,\qquad
\ad_h(v_k^*)=\omega_k v_k^*
\,.
$$
We then define the Lindblad generator
\begin{equation}
  \label{lho-f}
  L  = \sum_{k} \Big\{
  e^{\omega_k/2} \big( [v_k^*, \,\cdot\, ]v_k + v_k^* [\,\cdot\,,v_k] \big) 
  +e^{-\omega_k/2}
  \big( [v_k, \,\cdot\, ]v_k^* + v_k [ \,\cdot\, ,v_k^*] \big)\Big\}
  ,
\end{equation}
which is self-adjoint in $L_2(\sigma)$.

According to the standard terminology for Markov semigroups, see e.g. \cite[\S4.5]{Peres}, given $\epsilon\in(0,1/2)$, we define the \emph{$\epsilon$-mixing time} by
\begin{equation*}
  t_{\mathrm{mix}}(\epsilon)
  :=\sup_{\rho\in \mc A'_{+,1}}\inf\big\{t>0\colon \mathrm{d}_\mathrm{tr} (\rho P_t,\sigma)\leq\epsilon\big\},
\end{equation*}
where we recall that $\mc A'_{+,1}$ is the set of states on $\mathcal{A}$.
In words, for any $t\geq t_{\mathrm{mix}}(\epsilon)$ and any state $\rho$ the state $\rho P_t$ is $\epsilon$-close to $\sigma$ in trace distance.
The exponential ergodicity in trace distance of the QMS $(P_t)_{t\geq0}$ generated by $L$
is the content of the following result.

\begin{theorem}
Set $\Lambda=\inf_{k=1,\ldots,N} 2\varch\big(\frac{\omega_k}2\big)$.
Then
\begin{equation*}
4 \sup_{\rho \in   \mc A'_{+,1} }
\mathrm{d}_\mathrm{tr} (\rho P_t,\sigma)^2
\le (1+9e^{-2\Lambda(t-1)})^N-1.
\end{equation*}
In particular, for any $\epsilon\in(0,1/2)$
\begin{equation}\label{eq:stima-t-f}
t_{\mathrm{mix}}(\epsilon)\leq
1+
\frac{1}{2\Lambda}\log\Big(\frac{9N}{\log(1+4\epsilon^2)}\Big).
\end{equation}
\end{theorem}
The dependence on $N$ of the mixing time provided by \eqref{eq:stima-t-f} is optimal in general.
Indeed, in the case $\omega_k=0$, $k=1,\dots,N$,  by restricting the action of $(P_t)_{t\geq0}$ to diagonal matrices
we get a version of the classical Markov semigroup corresponding to the random walk on the hypercube $\{0,1\}^N$.
As follows from \S 18.2.2 and Theorem 20.3 in \cite{Peres}, the mixing time of this process  has the same $N$-dependence
as the right-hand side of \eqref{eq:stima-t-f}.
\begin{proof}
The spectral decomposition of $-L$ has been obtained in \cite{CM1}.
A complete system of eigenvector of $-L$,
parametrized by $\alpha=(\alpha_1,\dots,\alpha_N)\in(\{0,1\}\times\{0,1\})^N$,
is given by
$$
g_\alpha
=
g_{1,\alpha_1}\dots g_{N,\alpha_N}
\,,
$$
where
$$
g_{k,(0,0)}
=
\ind
\,,\,\,
g_{k,(1,0)}
=
a_k
\,,\,\,
g_{k,(0,1)}
=
a_k^*
\,,\,\,
g_{k,(1,1)}
=
e^{\omega_k/2} n_k
-
e^{-\omega_k/2}(\ind-n_k)
\,.
$$
The corresponding eigenvalue is
$$
\lambda_\alpha
=
2 \sum_{k=1}^n |\alpha_k| \varch\big(\tfrac{\omega_k}{2}\big)
\,,
$$
where $|\alpha_k|=i+j$ for $\alpha_k=(i,j)\in\{0,1\}\times\{0,1\}$.

By direct computation,
$$
\langle g_\alpha,g_\beta\rangle_\sigma
=
\delta_{\alpha,\beta}
\prod_{k=1}^N A(\omega_k)_{\alpha_k}
\quad\text{where}\quad
A(\omega_k)
=
\left(\begin{array}{cc}
1&\frac1{1+e^{\omega_k}} \\
\frac1{1+e^{-\omega_k}} & 1
\end{array}\right)
\,,
$$
while 
$$
\|g_\alpha\|
=
\prod_{k=1}^N B(\omega_k)_{\alpha_k}
\quad\text{where}\quad
B(\omega_k)
=
\left(\begin{array}{cc}
1&1 \\
1&e^{|\omega_k|/2}
\end{array}\right)
\,.
$$

Theorem \ref{t:tr<} and elementary computations yield
\begin{align*}
& 4\sup_{\rho \in   \mc A'_{+,1}}
\mathrm{d}_\mathrm{tr} (\rho P_t,\sigma)^2
\leq
\sum_{\alpha}e^{-2\lambda_\alpha t}\frac{\|g_\alpha\|^2}{\|g_\alpha\|^2_{\sigma}}
-1
\\
& =
\prod_{k=1}^N
\Big(
1+
8\varch^2(\omega_k/2)\varch(\omega_k)
e^{-4\varch(\omega_k/2)t}
+
e^{|\omega_k|}
e^{-8\varch(\omega_k/2)t}
\Big)
-1
\\
&\leq
\prod_{k=1}^N
\Big(
1+
9e^{2|\omega_k|-4\varch(\omega_k/2)t}
\Big)
-1
\\
& \le 
(1+9e^{-2\Lambda(t-1)})^N-1
\,,
\end{align*}
where, in the first inequality we used $\varch x\leq e^{|x|}$, 
while in the second inequality we used $|x|\leq\varch x$ and the definition of $\Lambda$.
\end{proof}

\section{Adjoint action of the Casimir element as a Lindblad generator}
\label{s:la}

In this section we analyze a family of QMS 
parametrized by a finite dimensional semisimple Lie algebra $\mf g$ and a
finite dimensional irreducible $\mf g$-module $V$.
As discussed in \cite{Ba,Li:Ju:La}, these QMS can be obtained by `transference' from the heat semigroup on the corresponding connected compact Lie group.
For this family we prove exponential ergodicity in trace distance
with a rate depending on $\mf g$ but not on $V$.

\subsection*{A motivating example}

Let $\ell_1,\ell_2,\ell_3$ be the orbital angular momentum of a
quantum particle in $\bb R^3$. In the Schr\"odinger representation they are
given by
$\ell_{j} = \epsilon_{jhk} \,x_h \,p_k=-i \, \epsilon_{jhk} \, x_h
\,\partial_{x_k}$ where $\epsilon_{jhk}$ is the totally antisymmetric
tensor of rank three and we used Einstein convention of summing over repeated indices. 
Note that $\ell_j$ are Hermitian and satisfy the
commutation relations $[\ell_j,\ell_h] = i \,\epsilon_{jhk}\, \ell_k$.
Denote by $S^2$ the two-dimensional sphere and
by $\Sigma$ the uniform probability on $S^2$.
%
Let also $\mc A$ be the family of bounded operators
on  $L_2(S^2; d\Sigma)$ endowed with the operator norm,
and consider the Lindblad operator on
$\mc A$ given by
\begin{equation*}
  L a 
  = \sum_{j=1}^3 \big\{ [\ell_j, a]\ell_j +\ell_j [a,\ell_j] \big\}
  =- \sum_{j=1}^3 [\ell_j,[\ell_j, a]]
\end{equation*}
that is well defined for a suitable dense subset of $\mc A$.
This operator can be seen as a non-commutative analogue of the
Laplace-Beltrami operator on $S^2$. 
%

The operator $L$ is not ergodic in $L_2(S^2; d\Sigma)$, its ergodic
decomposition is however simply achieved. Decompose
$L_2(S^2; d\Sigma)$ into the eigenspaces of the Casimir operator
$\ell^2=\sum_{j=1}^3 \ell_j^2$, i.e.\
$L_2(S^2; d\Sigma)=\bigoplus_{n=0}^{\infty} V_n$, where $V_n$, the
space of spherical harmonics of degree $n$, has dimension $2n+1$.  In
particular, $\ell^2$ acts on $V_n$ as the scalar operator
$n(n+1)\ind$.

Let $\mc A_n$ be the $C^*$-algebra $\mathrm{End}(V_n)$ endowed with
the operator norm.  It is simple to check that $L$ restricted to
$\mc A_n$ is reversible with respect to the normalized trace
$\sigma_n$ and ergodic.  The next natural issue regards the
speed of convergence to the invariant state $\sigma_n$ on each ergodic
component.  Relying on Theorem~\ref{t:tr<}, we will next show that
such convergence is in fact uniform in $n$. More precisely, denoting
by $\mc A'_{n,+,1}$ the set of states on $\mc A_n$, the following
bound holds.  There exists a universal constant $A\in (0,\infty)$
such that for any $n\in\bb Z_+$ and any $t\ge 0$
\begin{equation}\label{4t}
  \sup_{\rho\in \mc A'_{n,+,1}}
  \mathrm{d}_\mathrm{tr}\big(\rho P_t^n,\sigma_n\big)
  \le A \, e^{-2t}
 \end{equation}
where $P_t^n$ is the restriction of the semigroup generated by $L$
to $\mc A_n$.

\subsection*{Lie algebraic formulation}

Let $\mf g$ be a semisimple finite dimensional Lie algebra over $\bb C$
and let $\mf g_0$ be its compact real form, which has negative
definite Killing form $\kappa$, see e.g.\ \cite{Hum}. Then
$\mf g=\mf g_0 \otimes_{\bb R} \bb C $ is equipped with the
conjugation $(a_0+i b_0)^* := - a_0 + i b_0$ and the Euclidean inner
product $\langle a, b\rangle_{\mf g} := \kappa (a^*, b)$.  Let
$\pi\colon \mf g \to \mathfrak{gl} (V)$ be a finite dimensional
irreducible \emph{unitary representation} of $\mf g$ on $V$.  Namely,
$V$ is a Euclidean space over $\bb C$ with inner product
$(\cdot\,,\,\cdot)_V$ and $\pi$ is a Lie algebra
homomorphism such that $\pi(\mf{g}_0) \subset \mathfrak{u}(V)$. 
Here
$\mathfrak{gl} (V)$ denotes the general linear Lie algebra
$\mathrm{End} (V)$ with Lie bracket given by the commutator, while
$\mf{u}(V)\subset\mathrm{End}(V)$ is the real Lie algebra of
skew-Hermitian endomorphisms of $V$.  Observe that, by linearity,
$\pi(g^*)=\pi(g)^*$, $g\in\mf g$.  As customary in Lie theory, we will
typically drop $\pi$ from the notation.

Let $d:=\mathrm{dim}({\mf g})$ and fix an orthonormal basis
$\{\ell_j\}_{j=1}^{d}$ of the real Euclidean space $i\mf g_0$. 
Then $\{\ell_j\}_{j=1}^{d}$ is also basis of the complex space $\mf g$.
Let $\mc A$ be the $C^*$-algebra $\mathrm{End}(V)$ endowed with the
operator norm $\|\cdot\|$ and consider the Lindblad generator
$L \colon  \mc A \to \mc A$ given by
\begin{equation}
  \label{ling}
  L := - \sum_{j=1}^d \ad_{\ell_j}^2,
\end{equation}
which is reversible with respect to $\sigma$, the normalized trace
on $\mc A$. The operator $-L$ corresponds to the adjoint action of
the Casimir element of $\mf g$ on the $\mf g$-module $\mathrm{End}(V)$.
In particular, \eqref{ling} does not depend on the choice of the
orthonormal basis $\{\ell_j\}_{j=1}^d$.
Moreover, by the irreducibility of $V$, 
the QMS $(P_t)_{t\ge0}$ generated by $L$ is ergodic.

As in Section~\ref{s:2}, $L_2(\sigma)$ is the complex
Euclidean space $\mc A$ endowed with the inner product
\begin{equation}\label{eq:inner}
  \langle a, b \rangle_{\sigma} =\frac {\tr (a^* b)
  }{\mathrm{dim}(V)},
\end{equation}
i.e.\ $\langle \cdot\,,\,\cdot  \rangle_{\sigma}$ is the (normalized)
Hilbert-Schmidt inner product. Since $V$ is a unitary $\mf g$-module,
$L_2(\sigma)$ is also a unitary $\mf g$-module with respect to 
the adjoint action of $\mf g$.  

\begin{theorem}
  \label{t:mrl}
  There exists a constant $g_0\in(0,1]$ depending on the Lie
  algebra $\mf g$ but independent of the representation $V$ such that  
  $\gap(-L)\in [g_0,1]$.
  Furthermore, there is a constant $A\in (0,\infty)$ depending on
  $\mf g$ but independent of $V$  such that for any $t\ge 0$
  \begin{equation}\label{bound}
    \sup_{\rho\in \mc A'_{+,1}}
    \mathrm{d}_\mathrm{tr}\big(\rho P_t,\sigma\big)
    \le A \, e^{-\gap(-L)\,t}.
  \end{equation}
If $\mf g$ is simple, then the value of $g_0$ is given in the following table where $r\in\bb N$.
\begin{table}[h!]
\centering
\begin{tabular}{ c | c c c c c c c c c }
\vphantom{$\Big($}
$\mf g$ & $\mf{sl}_{r+1}$ & $\mf{so}_{2r+1}$ & $\mf{sp}_{2r}$ & $\mf{so}_{2r}$ & $E_6$ & $E_7$ & $E_8$ & $F_4$ & $G_2$ \\ 
\hline
\vphantom{$\Big($}
$g_0$ & 1 & $\frac{r}{2r-1}$ & $\frac{r}{r+1}$ & 1 & 1 & 1 & 1 & $\frac23$ & $\frac12$
\end{tabular}
\end{table}
\end{theorem}

For simple Lie algebras the bound $\gap(-L)\geq g_0$
is optimal in the sense that for each Lie algebra $\mf g$
there exists a representation $V$ for which $\gap(-L)$
coincides with the value of $g_0$ in the above table.
In fact, we have $\gap(-L)=g_0$ unless $\mf g$ is of type $B, C, F, G$,
and, even in these cases, $\gap(-L)=g_0$ except for ``few'' irreducible representations $V$.

The example discussed at the beginning of this section corresponds to the choice 
$\mf g = \mathfrak{so}_3 \simeq\mf{sl}_2$ together with its
unique irreducible representation $V_n$ of dimension $2n+1$. 
Observe that the factor $2$ on the right-hand side of \eqref{4t},
with respect to the value $g_0=1$ for $\mf g=\mf{sl}_2$ in Theorem \ref{t:mrl}
is due to a different normalization on the generators $\ell_j$, $j=1,2,3$.

\subsection*{Bounding operator norm by Hilbert-Schmidt norm}

In view of Theorem~\ref{t:tr<}, a key step for the proof of
Theorem~\ref{t:mrl} consists in obtaining a bound on the operator norm
of the eigenvectors of $L$ in terms of their norm in $L_2(\sigma)$.

Recall the following basic fact in harmonic analysis.
Let $G$ be a compact group that acts transitively on a set $X$ and endow
$X$ with the probability measure $\mu$ that is left invariant by the
action of $G$.  If $\mc B$ is a finite-dimensional $G$-stable subspace of
$L_2(X;\mu)$ then for any $b\in \mc B$  the bound
$\|b\|_\infty \le \sqrt{\mathrm{dim}(\mc B)} \, \|b\|_{L_2(X;\mu)}$ holds.  
The present aim, which has an independent
interest, is to provide a non-commutative version of this statement.


\begin{theorem}
  \label{t:ha}
  Let $\mc B$ be a $\mf g$-submodule of
  $L_2(\sigma)$. Then, for every $b\in\mc B$,
  \begin{equation}
    \label{op<hs}
    \|b\| \le \sqrt{\mathrm{dim}(\mc B)} \, \|b\|_\sigma
\,.  \end{equation} 
\end{theorem}

\begin{proof}
  Pick an orthonormal basis $\{b_{i}\}$ of $\mc B$.
  We claim that
  \begin{equation}
    \label{cl1}
    B:= \sum_i b_i^*b_i = \textrm{dim}(\mc B) \, \ind.
  \end{equation}
  Indeed, for $g\in\mf g$
  \begin{equation*}
    \begin{split}
    \sum_{i} [g,b_i^*] b_i  &= \sum_{i,j}
    \langle b_j^* ,[g,b_i^*] \rangle_\sigma b_j^* b_i
    \\
    &= \sum_{i,j}
    \langle b_i ,[b_j,g] \rangle_\sigma b_j^* b_i
    = \sum_{j} b_j^* [b_j,g]    
  \end{split}
\end{equation*}
  where we used that $\{b_i^*\}$ is an  orthonormal
  basis of $\mc B^*:=\{ b^*, \, b\in \mc B\}$ and $\mc B,\mc B^*$ are $\mf g$-submodules.
  We deduce that $[g,B]=0$. Hence, by the
  irreducibility of $V$ and Schur's lemma, $B=c \ind$ for some
  $c\in \bb R$.
  Moreover,
  \begin{equation*}
    c= \langle \ind, B \rangle_\sigma =\sum_{i}
    \langle b_i, b_i\rangle_\sigma= \textrm{dim}(\mc B)
  \end{equation*}
  as claimed.

  We next show that for each $u,v \in V$ we have
  \begin{equation}
    \label{s1}
    \sum_{i} \big| (u, b_i v)_V \big|^2 \le \textrm{dim}(\mc B) \, (u,u)_V\, (v,v)_V.
  \end{equation}
  Indeed, by Cauchy-Schwarz inequality and \eqref{cl1},
  \begin{equation*}
    \begin{split}
      \sum_{i} \big| (u, b_i v)_V \big|^2 & \le
    \sum_{i}  (u,u)_V (b_i v,b_i v)_V
    \\
    &
    =(u,u)_V\,  \Big( v\,,\, \sum_{i} b_i^*b_i v \Big)_V \, 
    = \textrm{dim}(\mc B) \, (u, u)_V \, (v,v)_V.
    \end{split}
  \end{equation*}
  
  We finally complete the proof by showing that the bound \eqref{s1}
  implies \eqref{op<hs}. Observe that if $b\in\mc B$ then
  $b =\sum_{i}\langle b_i,b\rangle_\sigma b_i$ and
  $\|b\| =\sup_{u,v} \big| (u,bv)_V\big|$ where the supremum is
  carried out over $u,v\in V$ such that $(u,u)_V=(v,v)_V=1$.  By
  Cauchy-Schwarz inequality we then get
  \begin{equation*}
    \begin{split}
    \|b\|^2 & = \sup_{u,v}  \big| (u,bv )_V\big|^2 
    =\sup_{u,v} \Big| \sum_{i} 
    (u, b_i v)_V 
     \langle b_i, b \rangle_\sigma
    \Big|^2
    \\
    & \le  \sup_{u,v} 
    \sum_{i} \big| (u,b_i v)_V\big|^2
    \,\sum_{i} \big|\langle b_i, b \rangle_\sigma \big|^2
    \le \textrm{dim}(\mc B) \, \|b\|_\sigma^2
  \end{split}
\end{equation*}
  where we used \eqref{s1} in the last step.
\end{proof}

\subsection*{Convergence in trace distance}

In this section we complete the proof of Theorem \ref{t:mrl}.
Referring the unfamiliar reader e.g.\ to \cite{Hum},
we recall some basic facts about finite dimensional representations of semisimple Lie algebras.

Let $(E,\Phi)$ be the \emph{root system} associated to the Lie algebra $\mf g$.
Here $E$ is a real Euclidean space of dimension $r=\rank(\mf g)$, with inner product $(\cdot\,|\,\cdot)$,
and $\Phi\subset E$ is a finite set of \emph{roots}.
Let also $\Delta=\{\alpha_1,\ldots,\alpha_r\}$ be a \emph{base}  of $\Phi$,
i.e. a basis of $E$ such that $\Phi=\Phi_+\sqcup(-\Phi_+)$,
where $\Phi_+=\Phi\cap E_+$, in which $E_+=\sum_i\bb R_+\alpha_i$ 
is the positive cone of $E$.
The elements of $\Delta$ are the \emph{simple roots}
and the elements of $\Phi_+$ are the \emph{positive roots}.
The collection of \emph{fundamental weights} $\Pi=\{\wp_1,\ldots,\wp_r\}\subset E$ 
is defined by $\frac{2(\alpha_i|\wp_j)}{(\alpha_i|\alpha_i)}=\delta_{i,j}$.
Finally, the set of \emph{dominant weights} is
$\Lambda^+
=
\big\{n_1\wp_1+\cdots+n_r\wp_r\big\}_{n_1,\ldots,n_r\in\bb Z_+}$.

A basic result in representation theory of Lie algebras is that 
the isomorphism classes of the irreducible finite dimensional representations of $\mf g$
are in one-to-one correspondence with $\Lambda^+$.
For $\lambda\in\Lambda^+$, we denote by $V_\lambda$ the corresponding 
irreducible representation.
The Weyl dimension formula then reads
\begin{equation}\label{dimeq}
\dim(V_\lambda)
=
\prod_{\alpha\in\Phi_+}\frac{(\lambda+\delta|\alpha)}{(\delta|\alpha)}
\,,
\end{equation}
where $\delta:=\frac12\sum_{\alpha\in\Phi_+}\alpha$.

By the Schur Lemma,
the Casimir element $C=\sum_{i=1}^d\ell_j^2$ acts as a scalar 
on each irreducible representation $V_\lambda$.
More precisely,
\begin{equation}\label{casimir}
C|_{V_\lambda}
=
c_\lambda\id_{V_{\lambda}}
\,,\qquad
c_\lambda
=
\frac{(\lambda|\lambda+2\delta)}{(\theta|\theta+2\delta)}
\,,
\end{equation}
where $\theta\in\Phi_+$ is the highest root, i.e.\
the dominant integral weight associated to the adjoint representation.
Indeed, the value by which the Casimir element acts 
on the irreducible representation $V_\lambda$ is 
$c_\lambda=c\,(\lambda,\lambda+2\delta)$,
where the constant $c$ does not depend on the representation $V_\lambda$,
see e.g. \cite{Hum}.
On the other hand, by the very definition of the Killing form and of the Casimir element $C$,
its value on the adjoint representation
is $c_\theta=1$,
whence $c=(\theta,\theta+2\delta)^{-1}$.

Given $\lambda_1,\lambda_2\in\Lambda^+$,
consider the decomposition of the tensor product $V_{\lambda_1}\otimes V_{\lambda_2}$
into its irreducible components, namely
\begin{equation}\label{tensor}
V_{\lambda_1}\otimes V_{\lambda_2}
=†
\bigoplus_{\mu\in\Lambda^+}V_\mu^{\oplus n_{\lambda_1,\lambda_2}(\mu)}
\,.
\end{equation}
Moreover, $n_{\lambda_1,\lambda_2}(\mu)=0$ unless
$\lambda_1+\lambda_2-\mu\in P^+
:=\big\{n_1\alpha_1+\cdots+n_r\alpha_r\big\}_{n_1,\ldots,n_r\in\bb Z_+}$.

\begin{lemma}\label{nuovo}
For each $\mu\in\Lambda^+$,
$$
\sup_{\lambda_1,\lambda_2\in\Lambda^+}\,\,
n_{\lambda_1,\lambda_2}(\mu)
\,\leq\, 
\dim(V_\mu)
\,.
$$
\end{lemma}
\begin{proof}
The claim follows directly from the PRV formula \cite[Thm.2.1]{PRV}.
\end{proof}

We are now ready to complete the proof of the convergence in trace distance.

\begin{proof}[Proof of Theorem \ref{t:mrl}]
Let $\lambda\in\Lambda^+$
be the dominant weight corresponding to the representation $V$,
i.e. $V\simeq V_\lambda$. 
Consider the canonical isomorphism
$\End V\simeq V\otimes V^* = V_\lambda\otimes V_{\lambda^*}$,
where $\lambda^*$ is the dominant weight associated to the dual representation $V_\lambda^*$.
The eigenspaces of the Lindblad generator $L$ in \eqref{ling} can then be obtained
from \eqref{tensor},
\begin{equation}\label{dec}
L_2(\sigma)
=
\bigoplus_{\mu\in\Lambda^+}V_\mu^{\oplus n_{\lambda,\lambda^*}(\mu)}
\,,
\end{equation}
and the corresponding eigenvalues can be read out of \eqref{casimir}:
\begin{equation}\label{eigenvalues}
-L|_
{V_\mu^{\oplus n_{\lambda,\lambda^*}(\mu)}}
=
\frac{(\mu|\mu+2\delta)}{(\theta|\theta+2\delta)}\id_{V_\mu^{\oplus n_{\lambda,\lambda^*}(\mu)}}
\,.
\end{equation}
In particular, $V_0$, the trivial 1-dimensional representation of $\mf g$,
appears with multiplicity $n_{\lambda,\lambda^*}(0)=1$,
and corresponds to $\bb C\ind$,
the eigenspace associated to the eigenvalue $0$ of $L$.
Since the adjoint representation is always present in the decomposition \eqref{dec},
we get $\gap(-L)\leq1$.
Moreover, by the observation following \eqref{tensor},
$\lambda+\lambda^*\in P^+$,
so that the dominant weights appearing in the right hand side of \eqref{dec}
must lie in $\Lambda^+\cap P^+$.
Hence, by \eqref{casimir}, we deduce that
\begin{equation}\label{eq:g0}
\mathrm{gap}(-L)
\geq
g_0
:=
\min_{\mu\in\Lambda^+\cap P^+\setminus\{0\}}
\frac{(\mu|\mu+2\delta)}{(\theta|\theta+2\delta)}
\,.
\end{equation}
This concludes the proof of the uniform lower bound on the spectral gap.

To prove the bound \eqref{bound}, we first observe that,
since  $\mathrm{d}_\mathrm{tr}(\rho,\sigma)\leq 1$,
we can assume $t\geq1$, 
By Theorem \ref{t:tr<} and \eqref{eigenvalues}, 
for $t\geq0$ and $\rho\in\mc A^\prime_{+,1}$,
$$
 \mathrm{d}_\mathrm{tr}\big(\rho P_t,\sigma\big)^2
 \leq
 \frac14
 \sum_{\mu\in\Lambda^+\cap P^+\setminus\{0\}}
 e^{-2\frac{(\mu|\mu+2\delta)}{(\theta|\theta+2\delta)}t}
\sum_{j}\| f_j^{(\mu)}\|^2
\,,
$$
where $f_j^{(\mu)}$ is an orthonormal basis of the eigenspace 
$V_\mu^{\oplus n_{\lambda,\lambda^*}(\mu)}$.
Let $g:=\gap(-L)$, by \eqref{eq:g0},
for $t\geq1$ the right hand side above is bounded by
$$
\frac14
 e^{-2gt}
  \sum_{\mu\in\Lambda^+\cap P^+\setminus\{0\}}
 e^{-2\big(
\frac{(\mu|\mu+2\delta)}{(\theta|\theta+2\delta)} 
-g
\big)}
\sum_{j}\| f_j^{(\mu)}\|^2
\,.
$$
By Theorem \ref{t:ha} we have 
$\| f_j^{(\mu)}\|^2\leq\dim(V_\mu)\| f_j^{(\mu)}\|_\sigma^2=\dim(V_\mu)$,
so that
$$
\sum_{j}\| f_j^{(\mu)}\|^2
\leq
n_{\lambda,\lambda^*}(\mu)\dim(V_\mu)^2
\leq
\dim(V_\mu)^3
\,,
$$
where we used Lemma \ref{nuovo}.
Combining the above bounds, we deduce \eqref{bound} with 
$$
A
=
\frac12
e^{g}
\Big(
\sum_{\mu\in\Lambda^+\cap P^+\setminus\{0\}}
e^{
-2\frac{(\mu|\mu+2\delta)}{(\theta|\theta+2\delta)}
}
\dim(V_\mu)^3
\Big)^{\frac12}
\,,
$$
which is finite thanks to \eqref{dimeq}.

It remains to show that, for simple Lie algebras, the value of $g_0$ as defined in \eqref{eq:g0} 
is the one in the statement.
%
%
This is achieved by a case by case analysis
that is carried out in Appendix \ref{App1}.
\end{proof}

\subsection*{Explicit eigenvectors of the Lindblad operator when
  $\mf g =\mf{sl}_2$}

In the special case when $\mf g=\mf{sl}_2\simeq\mf{so}_3$,
we provide explicit formulae for the eigenvectors of the Lindblad generator \eqref{ling}.
These can be viewed as a non-commutative version of the spherical harmonics.

Recall that the standard basis $\{e,f,h\}$ of $\mf{sl}_2$ has commutation relations
$[h,e]=2e,\,[h,f]=-2f,\,[e,f]=h$.
In this basis the Killing form is represented by the matrix
$$
\left(\begin{array}{ccc}
0&4&0 \\
4&0&0 \\
0&0&8
\end{array}\right)
$$
Hence, the dual basis of the standard basis with respect to the Killing form is
$\{\frac14 f,\frac14 e,\frac18 h\}$.
As a consequence, the Lindblad operator associated to an irreducible representation
$\pi:\,\mf{sl}_2\to\mf{gl}(V)$ is
\begin{equation}\label{L0sl2}
L
=
-\frac14\Big(
\ad_{\pi(e)}\ad_{\pi(f)}
+\ad_{\pi(f)}\ad_{\pi(e)}
+\frac12\ad_{\pi(h)}^2
\Big)
\,.
\end{equation}

Let $V_\lambda$ be the $n$-dimensional irreducible representation of $\mf{sl_2}$
(of highest weight $\lambda=n-1$).
An explicit realization of the action of the $\mf{sl}_2$-generators on $V_\lambda=\bb C^n$
is given by the following matrices
$$
e
=
\sum_{x=1}^n\sqrt{x(n-x)}E_{x,x+1}
\,,\quad
f
=
\sum_{x=1}^n\sqrt{x(n-x)}E_{x+1,x}
\,,\quad
h
=
\sum_{x=1}^n(n+1-2x)E_{x,x}
\,.
$$
The spectral decomposition of $L$ acting on $\mf{gl}(V)$ is (cf. \eqref{eigenvalues})
$$
\mf{gl}(V_\lambda)\simeq V_\lambda\otimes V_\lambda^*
=
\bigoplus_{i=0}^{n-1}V_{2i}
\,,\qquad
-L|_{V_{2i}}
=
\frac{i(i+1)}2
\,.
$$
An orthonormal basis of $V_{2i}\subset\End(V_\lambda)$
with respect to the inner product \eqref{eq:inner} is
$\{v^{(i)}_\ell\}_{\ell=0}^{2i}$, where
\begin{equation}\label{eq:vi}
v^{(i)}_\ell
=
\frac{i!\sqrt{n}}{\ell!\sqrt{\binom{2i}{\ell}\binom{n+i}{2i+1}}}
\sum_{x}
\gamma^{(i)}_{\ell}(x)
E_{x,x+i-\ell}
\,,
\end{equation}
in which
\begin{equation}\label{eq:gi}
\gamma^{(i)}_{\ell}(x)
=
\sqrt{\frac{(x-1)! (n-x-i+\ell)!}{(x-1+i-\ell)!(n-x)!}}
\,\,
\sum_{j=0}^\ell
(-1)^{\ell-j}
\binom{\ell}{j}\binom{x+i-j-1}{i}\binom{n-x+j}{i}
\,.
\end{equation}

\subsection*{Action of the Lindblad generator on diagonal matrices for $\mf g =\mf{sl}_2$}

If $\mf g=\mf{sl}_2$, then the generator $L$ in \eqref{L0sl2} preserves the Abelian subalgebra 
of diagonal matrices.
More precisely,
given $f:\,\{1,\dots,n\}\to\bb C$, by direct computation,
$$
L \sum_{x=1}^n f(x)E_{x,x}
=
\sum_{x=1}^n
(L_\mathrm{cl} f)(x) E_{x,x}
$$
where $E_{x,y}$ are the elementary matrices 
and $L_\mathrm{cl}$ is the classical Markov generator on $\{1,\dots,n\}$
given by
$$
(L_\mathrm{cl} f)(x)
=
\frac12
x(n-x)[f(x+1)-f(x)]
+
\frac12
(x-1)(n-x+1)[f(x-1)-f(x)]
\,,
$$
that is reversible with respect to the uniform probability on $\{1,\dots,n\}$.
Note that the process generated by $L_\mathrm{cl}$ is a discrete version of the Wright-Fisher
diffusion process.

The spectral decomposition of $-L_\mathrm{cl}$ can be obtained by restricting \eqref{eq:vi}-\eqref{eq:gi}
to diagonal matrices.
We have $\bb C^n=\bigoplus_{i=0}^{n-1}\bb C\gamma^{(i)}$,
where 
\begin{equation*}
\begin{split}
\gamma^{(i)}(x)
& =
\frac{\sqrt{n}}{\sqrt{\binom{2i}{i}\binom{n+i}{2i+1}}}
\sum_{j=0}^i
(-1)^{i-j}
\binom{i}{j}\binom{x+i-j-1}{i}\binom{n-x+j}{i}
\\
& =
i!\sqrt{\frac{(n-i-1)!}{(n+i)!}(2i+1)n}
\sum_{j=0}^i
(-1)^{i-j}
\binom{i}{j}\binom{x-1}{j}\binom{n-x}{i-j}
\end{split}
\end{equation*}
and $-L_\mathrm{cl}\gamma^{(i)}=\frac12i(i+1)\gamma^{(i)}$.
Note that, by the second expression, $\gamma^{(i)}(x)$ is a polynomial in $x$ of degree $i$.
In fact, these are the discrete Chebyshev polynomials, see e.g. \cite[\S 22.17]{AS}.

Finally, letting $P^\mathrm{cl}_t=\exp(tL_\mathrm{cl})$ be the Markov semigroup generated by $L_\mathrm{cl}$,
as a Corollary of Theorem \ref{t:mrl} we get that for any probability $\rho$ on $\{1,\dots,n\}$
\begin{equation*}
\mathrm{d}_\mathrm{TV}\big(\rho P^\mathrm{cl}_t,\sigma\big) \le A \, e^{-t}
\end{equation*}
where $\mathrm{d}_\mathrm{TV}$ is the total variation distance.
For the continuous Wright-Fisher diffusion semigroup this bound can be deduced from the exponential decay of the entropy which can be established by the Bakry-\'Emery criterion, see discussion below.
For the jump process generated by $L_\mathrm{cl}$, the criterion in \cite{Mi} yields a logarithmic Sobolev inequality with a constant uniform in $N$ but not sharp.
It is unclear to us wether the sharp value of the (modified) logarithmic Sobolev constant can be obtained for this jump process.

\subsection*{Lower bound on the spectral gap via Bakry-\'Emery criterion}

In the commutative case, a popular and easy to use criterion to deduce exponential ergodicity is the so-called Bakry-\'Emery curvature dimension condition, we refer to \cite{BGL} for an exhaustive reference on the subject.
As discussed there, without additional requirements, this criterion provides a lower bound on the spectral gap.
In the special case of diffusion semigroups, this criterion also yields a logarithmic Sobolev inequality, which imply the exponential convergence of the entropy.
Extensions of the Bakry-\'Emery theory to the non-commutative case are discussed in \cite{Br:Ga:Ju,CM1,CM2,DR,Li:Ju:La,Wi:Zh,Coreani}.
In the present setting of QMS with Lindblad generator given by the adjoint action of the Casimir element of a semi-simple Lie algebra, we here show that a uniform lower bound 
on the spectral gap can be obtained by using the Bakry-\'Emery curvature dimension condition.
Note that this argument does not yield the sharp values of the spectral gap computed in Appendix~\ref{App1} with the machinery of representation theory and detailed in Theorem \ref{t:mrl}.

Let the \emph{carr\`e du champ} $\Gamma\colon \mc A\times\mc A\to\mc A$ be defined by 
\begin{equation*}
\Gamma(a,b)
=
\frac12\Big(
L(ab)-(La)b-a(Lb)
\Big)
\end{equation*}
and the  \emph{carr\`e du champ it\'er\'e} $\Gamma_2\colon \mc A\times\mc A\to\mc A$ be defined by 
\begin{equation*}
\Gamma_2(a,b)
=
\frac12\Big(
L\Gamma(a,b)-\Gamma(La,b)-\Gamma(a,Lb)
\Big)
\,.
\end{equation*}

The following general statement is proven as in the commutative case,
see e.g. \cite[Prop.4.8.1]{BGL}.
\begin{proposition}\label{thm:ledu}
Let $L$ be an ergodic and reversible Lindblad generator.
Assume that there exists a constant $\gamma>0$ such that
\begin{equation}\label{eq:gamma}
\Gamma_2(a,a^*)\geq \gamma\, \Gamma(a,a^*)
\quad\text{ for all }\quad a\in\mc A\,.
\end{equation}
Then
$\gap(-L)\geq\gamma$.
\end{proposition}

\begin{theorem}
Let $L$ be the Lindblad generator in \eqref{ling}.
Then the inequality \eqref{eq:gamma} holds with $\gamma=\frac14$.
Hence, $\gap(-L)\geq\frac14$.
\end{theorem}
\begin{proof}
By a direct computation
$$
\Gamma(a,b)
=
\sum_{j}(\ad_{\ell_j}a)(\ad_{\ell_j}b)
\,,
$$
and
$$
\Gamma_2(a,b)
=
\sum_{i,j}(\ad_{\ell_i} \ad_{\ell_j}a)\,(\ad_{\ell_i}\ad_{\ell_j}b)
\,.
$$
By a symmetrization procedure,
$$
\Gamma_2(a,a^*)
=
\frac14\sum_{i,j}(\ad_{[\ell_i,\ell_j]} a)\,(\ad_{[\ell_i,\ell_j]}a^*)
+
\frac14
\sum_{i,j}(\{\ad_{\ell_i},\ad_{\ell_j}\}a)\,(\{\ad_{\ell_i},\ad_{\ell_j}\}a^*)
\,,
$$
where, as usual, $\{A,B\}=AB+BA$.
To conclude, we observe that the second term on the right hand side above is positive,
while the first term is equal to $\Gamma(a,a^*)/4$
since
$$
\sum_{i,j}[\ell_i,\ell_j] \otimes [\ell_i,\ell_j]
=
-\sum_j\ell_j\otimes\ell_j
\,,
$$
as follows from the facts 
that $\sum_j\ell_j\otimes\ell_j\in\mf g^{\otimes2}$ is left invariant by 
the adjoint action of $\mf g$,
and that the Casimir operator $\sum_i\ad_{\ell_i}^2$ acts as the identity on $\mf g$.
\end{proof}

\section{Bose Ornstein-Uhlenbeck semigroup}
\label{s:qou}

In this section we consider the Bose Ornstein-Uhlenbeck semigroup corresponding to $N$ free bosons and discuss the spectral decomposition of the corresponding Lindblad generator together with its exponential ergodicity in trace distance.
In contrast to the previous examples, the underlying Hilbert space is infinite dimensional.
According to the terminology in \cite[\S3]{Fa}, in such a context a QMS is defined on a von Neumann algebra  rather than a $C^*$-algebra.
In particular $(P_t)_{t\geq0}$ is required only to be weakly* continuous and not strongly continuous.
For the present purposes it is however convenient to realize the Bose Ornstein-Uhlenbeck semigroup as a strongly continuous semigroup on the $C^*$-algebra of compact operators on a suitable Hilbert space.
In the commutative case, this choice corresponds to the realization of Markov semigroups as strongly continuous semigroups on the commutative $C^*$-algebra given by the continuos functions vanishing at infinity on some metric space.
While the Bose Ornstein-Uhlenbeck semigroup depends only on the canonical commutation relationships between creation and annihilation operators, it will be convenient, as in \cite{CFL}, to realize this operator in the Hilbert space in which the bosonic number operators are diagonal.

Fix a postive integer $N$.
In this section, we employ the usual multi-index notation.
For $\ell,j \in \bb Z_+^N$ and $z=(z_1,\ldots, z_N)\in\bb C^N$,  we denote $|\ell|:= \ell_1+\cdots +\ell_N$,
$\ell!:=\ell_1!\cdots \ell_N!$, $j\le \ell$ if $j_1\le \ell_1$,\dots,
$j_N\le \ell_N$, $z^\ell := z_1^{\ell_1}\cdots z_N^{\ell_N}$, and
$\binom {\ell}{j} := \binom {\ell_1}{j_1}\cdots \binom {\ell_N}{j_N}$.

Consider the Hilbert space $H=\boldsymbol{\ell}_2(\bb Z_+^N)$ and denote by $\{|\ell\rangle\}_{\ell\in\bb Z_+^N}$ its canonical  orthonormal basis.
The bosonic number operators $n_k$, $k=1,\dots,N$, are the self-adjoint pairwise commuting operators defined by $n_k |\ell\rangle=\ell_k |\ell\rangle$, with domain
\begin{equation*}
  \Dom(n_k)=\Big\{v=\sum_\ell v_\ell |\ell\rangle\in H\colon \sum_\ell \ell_k^2|v_\ell|^2<\infty\Big\}.
\end{equation*}
The annihilation and creation operators $a_k$, $a_k^*$ on $H$, with domain $\Dom(a_k)=\Dom(a_k^*)=\Dom(\sqrt{n_k})$, $k=1,\dots,N$ are defined by
\begin{equation}\label{eq:aastar}
  a_k |\ell\rangle =
  \sqrt{\ell_k} \: |\ell- {e}_k\rangle,
  \qquad a_k^* |\ell\rangle = \sqrt{\ell_k+1}\: |\ell+e_k \rangle
\end{equation}
where $e_k\in \bb Z_+^N$ has coordinates $(e_k)_{k'} := \delta_{k,k'}$.
We understand that  $a_k |\ell\rangle=0$ when $\ell_k=0$.
The operators $a_k$, $a_k^*$ are closed, mutually adjoint, and they satisfy $a_k^*a_k=n_k$, $a_ka_k^*=n_k+\ind$.

Fix a family $\{\omega_k\}_{k=1}^N$ of strictly positive reals.
We hereafter assume that they are bounded away from zero uniformly in $N$.
Namely, there exists $\Lambda\in(0,\infty)$ independent on $N$ such that $2\varsh(\omega_k/2)\geq\Lambda$, $k=1,\dots,N$.
The corresponding free boson Hamiltonian is $h=\sum_k\omega_kn_k$ which is self-adjoint on $\Dom(h)=\bigcap_k\Dom(n_k)$.
We denote by $\mathcal{B}$ the $C^*$-algebra of bounded operators on $H$ endowed with the operator norm $\|\cdot\|$.
Let also $\mathcal{K}$ be the $C^*$-subalgebra of $\mathcal{B}$ given by the compact operators on $H$, namely the $\|\cdot\|$ closure of finite rank operators on $H$.
Note that $\mathcal{B}$ is unital while $\mathcal{K}$ is not.

Let $\sigma$ be the positive trace class operator on $H$ given by $\sigma=e^{-h}/\tr(e^{-h})$.
We regard $\sigma$ also as a state on $\mathcal{K}$ by setting $\sigma(f)=\tr(\sigma f)$.
We have absorbed the dependence on the inverse temperature on the frequencies $\omega_k$.
Denote by $\mathcal{D}^{(0)}$ the linear span of $\{|\ell\rangle\langle\ell'|\}_{\ell,\ell'\in\bb Z_+^N}$, which is a dense subset of $\mathcal{K}$ invariant by left and right multiplication by $a_k$, $a_k^*$, $k=1,\dots,N$.
By direct computation, on $\mathcal{D}^{(0)}$
\begin{equation*}
  \ad_h(a_k) = [h,a_k]=-\omega_k a_k,\qquad \ad_h(a_k^*) = [h,a_k^*]=\omega_k a_k^*.
\end{equation*}

We then consider the operator $L^{(0)}$ on $\mathcal{K}$ defined on $\mathcal{D}^{(0)}$ by
\begin{equation}
  \label{lho-b}
  L^{(0)}  = \sum_{k=1}^N \Big\{
  e^{\omega_k/2} \big( [a_k^*, \,\cdot\, ]a_k + a_k^* [
  \,\cdot\,,a_k] \big) 
  +e^{-\omega_k/2}
  \big( [a_k, \,\cdot\, ]a_k^* + a_k [ \,\cdot\, ,a_k^*] \big)\Big\}.
\end{equation}
As follows from \cite[Thm 5.1]{CFL}, the graph norm closure of $L^{(0)}$ generates a strongly continuous contraction semigroup $(P_t)_{t\ge0}$ on $\mc K$ (as a Banach space) such that $P_t$ is a completely positive operator on $\mathcal{K}$ for any $t\geq0$.
Furthermore, $\mathcal{D}^{(0)}$ is a core for the generator of $(P_t)_{t\ge0}$ and $\sigma$ is a reversible state.
More precisely, in \cite{CFL} the semigroup generated by $L^{(0)}$ is constructed first on the Hilbert space associated to the KMS inner product induced by the state $\sigma$ by Dirichlet form techniques, and then it is shown that it has the Feller property, i.e., it preserves $\mc K$.
In view of \cite[Thm.2.9]{CM1}, the notion of reversibility used in \cite{CFL} is equivalent to the one employed here based on the GNS inner product induced by $\sigma$.

We denote by $L_2(\sigma)$ the Hilbert space obtained by completing $\mathcal{K}$ with respect to the Euclidean distance induced by the GNS inner product $\langle f,g\rangle_\sigma=\sigma(gf^*)$.
We regard $\mathcal{K}$ as a dense subset of $L_2(\sigma)$.
By the Kadison-Schwarz inequality, $(P_t)_{t\ge0}$ extends to a strongly continuous self-adjoint contractions semigroup on $L_2(\sigma)$ that is denoted by $(P_t^{(2)})_{t\ge0}$.
Let finally $L^{(2)}$ with domain $\Dom(L^{(2)})$ be its self-adjoint generator.
Observing that $(L^{(0)},\mathcal{D}^{(0)})$ is symmetric when regarded as an operator on $L_2(\sigma)$, the generator $L^{(2)}$ can also be obtained, as stated in the next lemma, by taking the graph-norm closure in $L_2(\sigma)$ of $(L^{(0)},\mathcal{D}^{(0)})$.

\begin{lemma}
  Regarding $(L^{(0)},\mathcal{D}^{(0)})$ as densely defined operator on $L_2(\sigma)$, its graph-norm closure is $(L^{(2)},\Dom(L^{(2)}))$.
\end{lemma}

\begin{proof}
  By direct computation $(-L^{(0)},\mathcal{D}^{(0)})$ is a positive symmetric operator on $L_2(\sigma)$.
  On the other hand, by its very definition, $(-L^{(2)},\Dom(L^{(2)}))$ is a positive self-adjoint operator on $L_2(\sigma)$ that extends  $(-L^{(0)},\mathcal{D}^{(0)})$.
  It is therefore enough to show that  $(-L^{(0)},\mathcal{D}^{(0)})$ has a unique self-adjoint extension.
  Since the graph norm closure in $\mathcal{K}$ of $(L^{(0)},\mathcal{D}^{(0)})$ generates a strongly continuous contraction semigroup on $\mathcal{K}$, the Lumer-Philips Theorem, see e.g. \cite[Thm.X.48]{RS2}, implies that for each $\lambda>0$ the set $(\lambda-L^{(0)})\mathcal{D}^{(0)}$ is dense in $\mathcal{K}$.
  Since $\mathcal{K}$ is dense in $L_2(\sigma)$ we get that for each $\lambda>0$ the set $(-\lambda+L^{(0)})\mathcal{D}^{(0)}$ is dense in $L_2(\sigma)$.
  By the characterization of self-adjoint extensions of positive symmetric operators in \cite[Thm.X.2]{RS2}, we deduce that $(-L^{(0)},\mathcal{D}^{(0)})$ has a unique self-adjoint extension.
\end{proof}

As next stated, the action of $L^{(2)}$ on polynomials in the annihilation and creation operators is still given by the informal expression on the right hand side of \eqref{lho-b}.

\begin{lemma}\label{lem:ide}
  Let $\mathcal{W}$ be the unital algebra generated by the operators $a_k$, $a_k^*$, $k=1,\dots,N$ in which we understand that the elements of $\mathcal{W}$ are defined on the linear span of $\{|\ell\rangle\}_{\ell\in\bb Z_+^N}$.
  Then $\mathcal{W}$ can be identified with a subset of $L_2(\sigma)$.
  Furthermore, under this identification, $\mathcal{W}\subset\Dom(L^{(2)})$ and for each $w\in\mathcal{W}$
  \begin{equation}\label{eq:Lw}
      L^{(2)}w  = \sum_{k=1}^N \Big\{
  e^{\omega_k/2} \big( 2a_k^*w a_k - w a_k^*a_k-a_k^*a_kw \big) 
  +e^{-\omega_k/2}
  \big( 2a_kw a_k^* - w a_ka_k^*-a_ka_k^*w \big)\Big\}. 
  \end{equation}
\end{lemma}

\begin{proof}
  Given $w\in\mathcal{W}$ and $\alpha\in\bb N$, let $w_\alpha\in\mathcal{D}^{(0)}$ be the operator defined by
  \begin{equation*}
    w_\alpha|\ell\rangle
    =
    \begin{cases}
      w|\ell\rangle & \text{if $\ell_k\leq\alpha$ for all $k=1,\dots,N$,}\\
      0 & \text{otherwise.}
    \end{cases}
  \end{equation*}
  By direct computation, the sequence $\{w_\alpha\}$ is Cauchy with respect to $\|\cdot\|_\sigma$.
  The first statement is achieved by identifying $w$ with the equivalence class of the Cauchy sequence $\{w_\alpha\}$.

  To prove the second statement, fix $w\in\mathcal{W}$ and let $\{w_\alpha\}$ be as in the previous displayed equation.
  Since, as observed before, $L^{(2)}$ can be obtained as the graph norm closure of $(L^{(0)},\mathcal{D}^{(0)})$ in $L_2(\sigma)$, it is enough to show that the sequence $\{L^{(0)}w_\alpha\}$ is Cauchy in $L_2(\sigma)$ and in the same equivalence class of the Cauchy sequence that represents the right  hand side of \eqref{eq:Lw}.
  Again, this is achieved by a direct computation.
\end{proof}

In order to describe the spectral decomposition of $-L^{(2)}$,
given $j,m\in \bb Z_+^N$, let $g_{j,m}\in \mathcal{W}\subset\Dom(L^{(2)})$ be given by
\begin{equation}
  \label{plm}
  \begin{split}
  g_{j,m} &:=\sum_{i\in {\bb Z}_+^N}  i!\, \binom{j} i
   \binom{m} i (-1)^{|i|} \gamma^i (a^*)^{j -i}a^{m -i}
   \\
   & =\sum_{i\in {\bb Z}_+^N}  i!\, \binom{j} i
   \binom{m} i (-1)^{|i|} \delta^i a^{m -i}(a^*)^{j -i}
 \,,
 \end{split}
 \end{equation}
 where $a=(a_1,\ldots,a_N)$, $a^*=(a_1^*,\ldots,a_N^*)$,
 $\gamma =(\gamma_1,\ldots,\gamma_N)$,
 $\delta=(\delta_1,\ldots,\delta_N)$, 
 and
\begin{equation}
  \label{al=}
  \gamma_k = \frac {e^{-\omega_k/2}}{2\varsh(\omega_k/2)}
  ,\qquad
  \delta_k = \frac {e^{\omega_k/2}}{2\varsh(\omega_k/2)}.
\end{equation}

\begin{theorem}\label{thm:specdec}
 The operator $-L^{(2)}$ has purely discrete spectrum given by
    \begin{equation}
    \label{spl}
    \lambda_j := 2 \sum_{k=1}^N  
    \varsh\big(\tfrac {\omega_k}2\big) j_k,
    \qquad j\in\bb Z_+^N.
  \end{equation}
 If $\varsh(\omega_k/2)$, $k=1,\dots,N$, are rationally independent,
  the eigenspace $\ms U_j$ associated to $\lambda_j$
  has dimension $(j_1+1) \cdots (j_N+1)$.
  An orthonormal basis of $\ms U_j$ is given by
  $f_{j,m} := g_{j-m,m}/\| g_{j-m,m}\|_{\sigma}$,
  $m\le j$.
  The normalization is
  \begin{equation}
    \label{aul}
    \| g_{j,m}\|_{\sigma}^2= j!\, m!\, \gamma^j\, \delta^{m}
    ,\qquad j,m\in\bb Z_+^N.
  \end{equation}
For arbitrary $\omega_k$,
the eigenspace associated to the eigenvalue $\lambda$ is
  $\mc U_\lambda:=\bigoplus_{j\,:\,\lambda_j=\lambda}\mc U_j$.
\end{theorem}

We point out that the spectrum of the Lindblad generator of the QMS $(P_t)_{t\geq0}$ on $\mathcal{K}$ is different from the spectrum of the operator $L^{(2)}$ described in the above theorem.
We refer to \cite{Meta} for a discussion on this topic in the commutative setting.
Postponing the proof of Theorem~\ref{thm:specdec}, we next discuss the convergence
to equilibrium of the Bose Ornstein-Uhlenbeck semigroup.  To this end,
let $\rho$ be a state on $\mathcal{K}$.  By \cite[Thm.VI.26]{RS1},
there exists positive trace class operator on $H$ with unit trace,
still denoted by $\rho$, such that $\rho(f)=\tr(\rho f)$,
$f\in\mathcal{K}$.  In order to formulate the exponential ergodicity
in trace distance of $(P_t)_{t\geq0}$, we need to impose some
conditions on the initial state $\rho$.  Indeed, in this infinite
dimensional case, exponential ergodicity in trace distance cannot hold
uniformly with respect to $\rho$.
To this end, we denote by $\mc S$ the set of states on $\mc K$ 
for which  the number operators $n_k$ have finite moments of any
order. Namely, $\rho$ belongs to $\mc S$ if for each $j\in\bb Z_+^N$
\begin{equation}\label{eq:moment}
  \rho( n^j ) :=
  \sum_{\ell\in\bb Z_+^N}\ell^j\langle\ell|\rho|\ell\rangle<\infty.
\end{equation}
Readily, the thermal state $\sigma$ belongs to $\mc S$.
Given $K\in (0,\infty)$, let finally $\mc S_{K}$ be the set of
$\rho\in \mc S$ such that 
\begin{equation}
  \label{mcK}
  \rho\big( n^j \big)
  = \sum_{\ell\in\bb Z_+^N}\ell^j\, \langle\ell|\rho|\ell\rangle
  \le  j! \: K^{|j|},
  \qquad \forall \, j\in \bb Z_+^N.
\end{equation}
Given $\epsilon\in(0,1/2)$, we define the \emph{$\epsilon$-mixing time in $\mathcal{S}_K$} by
\begin{equation*}
  t_{\mathrm{mix}}(\epsilon,\mathcal{S}_K)
  :=\sup_{\rho\in\mathcal{S}_K}\inf\big\{t>0\colon \mathrm{d}_\mathrm{tr} (\rho P_t,\sigma)\leq\epsilon\big\},
\end{equation*}
so that for  any $t\geq t_{\mathrm{mix}}(\epsilon,\mathcal{S}_K)$ and any $\rho\in\mathcal{S}_K$ the state $\rho P_t$ is $\epsilon$-close to $\sigma$ in trace distance.
Recalling that $\Lambda$ is a uniform lower bound on 
$2\varsh(\omega_k/2)$, as next stated, the semigroup $(P_t)_{t\geq0}$ is 
exponentially ergodic in trace distance uniformly on $\mc S_K$. 

\begin{theorem}\label{thm-bose}
For each $K\in (0,\infty)$ there is a constant
$A=A(\Lambda,K)$ such that for any $N\in\bb N$ and  $t\ge 0$ 
  \begin{equation}\label{label4}
4 \sup_{\rho \in   \mc S_K }
    \mathrm{d}_\mathrm{tr} (\rho P_t,\sigma)^2
     \le 
\big(
 1+Ae^{-2\Lambda t}
\big)^N -1
\,.
   \end{equation}
In particular, for any $\epsilon\in(0,1/2)$
\begin{equation}\label{eq:stima-t-b}
t_{\mathrm{mix}}(\epsilon,\mathcal{S}_K)\leq
\frac{1}{2\Lambda}\log\Big(\frac{AN}{\log(1+4\epsilon^2)}\Big).
\end{equation}
\end{theorem}

The proof of Theorem \ref{thm:specdec} is achieved by the next two lemmata.
In the former, which is purely algebraic, we show that each $f_{j,m}$ is eigenvector of $-L^{(2)}$ with eigenvalue $\lambda_j$.
In the latter we show that the linear span of $\{f_{j,m}\}_{m\leq j}$ is dense in $L_2(\sigma)$.

\begin{lemma}
  As a formal power series, 
   \begin{equation}
   \label{gen=}
     g(z,w):=
     \sum_{j,m\in\bb Z_+^N}\frac{{z}^j w^m}{j!m!} 
     g_{j,m}
   \end{equation}
   satisfies
       \begin{equation}
      \label{elgf}
      -L^{(2)} g (z,w) =
      2 \sum_{k=1}^N   \varsh\big(\tfrac {\omega_k}2\big)
      \Big(z_k  \frac {\partial}{\partial z_k} +
      w_k  \frac {\partial}{\partial w_k}\Big) g(z,w)
    \end{equation}
    and
    \begin{equation}
  \label{cos}
  \langle g(\tilde z,\tilde w), g(z,w) \rangle_\sigma 
  = e^{ \sum_{k}\big(\gamma_k {\overline{\tilde z}}_k z_k
  + \delta_k {\overline{\tilde w}}_k w_k\big)}.
\end{equation}
\end{lemma}

\begin{proof}
  In view of \eqref{plm}, a direct computation (at  the level of formal power series) yields
  \begin{equation*}
    g(z,w)=
     e^{- \sum_{k} \gamma_k z_kw_k } \,e^{z a^*}
     \,e^{w a}
     =e^{-\sum_{k} \delta_k z_k w_k} \,e^{w a}   \,e^{z a^*}.
   \end{equation*}
      Recalling \eqref{lho-b} and applying the Leibniz rule,
    \begin{equation*}
      \begin{split}
        &  -L^{(2)} g (z,w)
        \\
        &=
      \sum_{k=1}^N \Big[
      e^{\omega_k/2} \big(
      z_k a_k^* g(z,w)  + g(z,w) w_k a_k \big)
      -e^{-\omega_k/2} \big(
      g(z,w) z_k a_k^*  + w_k a_k g(z,w)\big)\Big]
      \\
      &=  2 \sum_{k=1}^N
            \varsh\big(\omega_k/2\big)
      \big( z_k a_k^* g(z,w)  + g(z,w) w_k a_k \big)
      -2 e^{-\omega_k/2} z_k w_k g(z,w)
    \end{split}
    \end{equation*}
    where we used $[a_k,g(z,w)]=z_k g(z,w)$ and $[a_k^*,g(z,w)]=-w_k g(z,w)$.
    Moreover, by a direct computation,
    \begin{equation*}
      \begin{split}
        &
        2 \sum_{k=1}^N   \varsh\big(\tfrac {\omega_k}2\big)
      \Big(z_k  \frac {\partial}{\partial z_k} +
      w_k  \frac {\partial}{\partial w_k}\Big) g(z,w)
      \\
      &\;=
      2 \sum_{k=1}^N
            \varsh\big(\omega_k/2\big)
            \big( z_k a_k^* g(z,w)  + g(z,w) w_k a_k \big)
            - 4 \sum_{k=1}^N \varsh\big(\omega_k/2\big)    
      \gamma_k z_k w_k g(z,w).
    \end{split}
    \end{equation*}
Hence the identity \eqref{elgf} follows from \eqref{al=}.

In order to prove \eqref{cos}, we first observe that the canonical commutation relations of $a_k$, $a_k^*$, $k=1,\dots,N$ imply
\begin{equation*}
  e^{wa} e^{za^*} = e^{\sum_k z_kw_k}  e^{za^*} e^{wa}. 
\end{equation*}
This identity and the definition \eqref{gen=} yield
\begin{equation}
  \label{tr1}
  \begin{split}
    &  \tr \big( g(\tilde z,\tilde w)^* e^{- h} g(z,w) \big)
    = e^{-\sum_k \delta_k (z_kw_k
      +{\overline{\tilde z}}_k {\overline{\tilde w}}_k)}
    \tr \big( e^{{\overline{\tilde z}}a} e^{{\overline{\tilde w}}a^*}
    e^{- h} e^{wa} e^{za^*} 
    \big) 
    \\
    &\quad
    = e^{-\sum_k [\delta_k (z_kw_k
      +{\overline{\tilde z}}_k {\overline{\tilde w}}_k)
      + z_k{\overline{\tilde z}}_k]
    }
    \tr \big( 
    e^{(z+{\overline{\tilde w}}) a^*}
    e^{- h} e^{(w+{\overline{\tilde z}}) a}
    \big).
  \end{split}
\end{equation}
Computing the above trace in the  basis $\{|\ell\rangle\}_{\ell\in\bb Z_+^N}$, using \eqref{eq:aastar} and expanding the exponentials, we get
\begin{equation*}
  \begin{split}
    \tr \big(e^{za^*}e^{-h} e^{wa}\big)
    &= \sum_{\ell\in \bb Z_+^N}
    \langle\ell| e^{za^*}e^{- h} e^{wa}  |\ell\rangle
    \\
    &
    =
    \sum_{\ell,m\in \bb Z_+^N} \binom{\ell}{m} \frac {z^m w^m}{m!} 
    e^{-\sum_{k}\omega_k(\ell_k-m_k+1/2)}
    \,.
   \end{split}
\end{equation*}
Since
\begin{equation*}
  \sum_{\ell \in \bb Z_+^N}  \binom{\ell}{m} x^\ell =
  \prod_{k=1}^N \frac {x_k^{m_k}}{(1-x_k)^{m_k+1}},\qquad x=(x_1,\cdots,x_N)\in \bb C^N, 
\end{equation*}
we have
\begin{equation*}
  \begin{split}
    \tr \big(e^{za^*}e^{-h} e^{wa}\big)
    &=    \prod_{k=1}^N  \Big\{ \frac{e^{- \omega_k/2}}
    {1-e^{-\omega_k}}
    \sum_{m=0}^\infty \frac {1}{m!} 
    \Big( \frac {z_k w_k}{ 1-e^{-\omega_k }} \Big)^m
    \Big\}
    \\
    & =
    \exp\Big\{\sum_{k} \Big[
    \frac {z_k w_k}{ 1-e^{-\omega_k }}
    - \log \big( 2\varsh(\omega_k/2) \big)\Big]
    \Big\}.
  \end{split}
\end{equation*}
Noticing that $\delta_k^{-1} =1-e^{-\omega_k}$, 
and plugging the above result in \eqref{tr1}, we deduce
\begin{equation*}
   \langle g(\tilde z,\tilde w), g(z,w) \rangle_\sigma 
   = e^{-\sum_k [\delta_k (z_kw_k
      +{\overline{\tilde z}}_k {\overline{\tilde w}}_k)
      + z_k{\overline{\tilde z}}_k]
    }
    \, e^{\sum_{k} 
   \delta_k {(z_k+{\overline{\tilde w}}_k)(w_k+{\overline{\tilde z}}_k)} }
\end{equation*}
which implies \eqref{cos} as $\delta_k-1=\gamma_k$.
\end{proof}

\begin{lemma}\label{lem:denso}
Under the identification in Lemma~\ref{lem:ide}, the set $\mathcal{W}$ is dense in $L_2(\sigma)$.
\end{lemma}
\begin{proof}
We follow the argument in \cite[Lem.7.1]{CFL}.
For simplicity of notation, we prove the statement in the case $N=1$.
Since $\mc D^{(0)}$ is dense in $L_2(\sigma)$, it is enough to show that for $b\in\mc D^{(0)}$
\begin{equation}\label{eq:claim}
\langle (a^*)^ia^j,b\rangle_\sigma=0
\,\,\text{ for all }\,\, i,j\in\bb Z_+
\quad\Longrightarrow\quad
b=0.
\end{equation}
Fix $i,j\in\bb Z_+$.
By direct computation, for $b=\sum_{\ell,m}b_{\ell,m}|\ell\rangle\langle m|\in\mc D^{(0)}$,
we have
\begin{align*}
\langle (a^*)^ia^{i+j},b\rangle_\sigma
& =
\frac1{1-e^{-\omega_1}}
\sum_{\ell,m}
b_{\ell,m}
e^{-\omega_1\ell}
\langle m|
(a^*)^{i+j}a^i
|\ell\rangle \\
& =
\frac1{1-e^{-\omega_1}}
\sum_{\ell}
b_{\ell,\ell+j}
e^{-\omega_1\ell}
\langle \ell+j|
(a^*)^{i+j}a^i
|\ell\rangle
\,.
\end{align*}
Since
$$
\langle \ell+j|
(a^*)^{i+j}a^i
|\ell\rangle
=
\ell(\ell-1)\cdots(\ell-i+1)
\langle \ell+j|
(a^*)^{j}
|\ell\rangle
\,,
$$
the polynomial
$$
b^{(j)}(z)
=
\frac1{1-e^{-\omega_1}}
\sum_{\ell}
b_{\ell,\ell+j}
e^{-\omega_1\ell}
\langle \ell+j|
(a^*)^{j}
|\ell\rangle
z^\ell
$$
satisfies
$$
\frac{d^ib^{(j)}}{dz^i}\Big|_{z=1}
=
\langle (a^*)^ia^{i+j},b\rangle_\sigma
\,.
$$
Hence, the assumption in \eqref{eq:claim} implies $b^{(j)}(z)=0$,
which yields $b_{\ell,\ell+j}=0$ for every $\ell,j\in \bb Z_+$. 
Exchanging the roles of $i$ and $i+j$ and of $\ell$ and $m$,
we also get $b_{\ell+j,\ell}=0$, completing the proof of \eqref{eq:claim}.
\end{proof}

\begin{proof}[Proof of Theorem \ref{thm-bose}]
  Within the present infinite dimensional context, we define the
  quantum $\chi^2$-divergence $\mathrm{D}_{\chi^2} (\cdot | \sigma)$ as the map on the set of states on $\mathcal{K}$ defined by the variational formula
  \eqref{dc2var}, namely
  \begin{equation}\label{eq:dichi}
   \mathrm{D}_{\chi^2} (\rho | \sigma) :=
  \sup_{g\in \mc K \,\colon\, \|g\|_\sigma =1}  \big| \rho(g) -\sigma(g)
  \big|^2. 
  \end{equation}
  We then notice that the variational proof of Lemma~\ref{tr<c2}
  presented in Section~\ref{s:2} applies also to the current setting.

  We next claim that for each $K\in(0,\infty)$ there exists $T_K$ such that the following holds.
  For each $t\geq T_K$, $g\in \mc K$, and $\rho\in S_K$ we have
  \begin{equation}
    \label{cla}
    \rho\big(P_t g\big) =
    \sum_{j\in \bb Z_+^N} \sum_{m\le j}
    e^{-\lambda_j t} \: \big\langle f_{j,m}, g \big\rangle_\sigma
    \: \rho\big(  f_{j,m}\big)
  \end{equation}
  where $\{f_{j,m}\}$ are the eigenvectors of $L^{(2)}$ provided by
  Theorem~\ref{thm:specdec} and $\lambda_j$ the corresponding
  eigenvectors.
  Note that $\rho\big(f_{j,m}\big)$ is well defined as $\rho\in \mc S_K$ and $f_{j,m}\in\mathcal{W}$. 

  Postponing the proof of the equality in \eqref{cla}, we first show that the (numerical) series on its right hand side is absolutely convergent.
  By Cauchy-Schwarz inequality, \eqref{spl}, elementary estimates, and choosing $T_K$ large enough, this follows from the following bound.
  There exists $C=C(K,\omega)\in(0,\infty)^N$ such that
  \begin{equation}
    \label{eq:bdf}
    \big|\rho\big(f_{j,m}\big)\big|
    \leq C^j
    \qquad j,m\in\bb Z_+^N, m\leq j.
  \end{equation}
  As follows from Theorem~\ref{thm:specdec}, this bound is equivalent to
  \begin{equation}
        \label{eq:bdf1}
    \big|\rho\big(g_{\ell,m}\big)\big|
    \leq \sqrt{\ell!m!\gamma^{\ell}\delta^{m}}C^{\ell+m}
    \qquad \ell,m\in\bb Z_+^N.
  \end{equation}
  According to \eqref{plm},
  \begin{equation}
    \label{eq:bdf2}
    \big|\rho\big(g_{\ell,m}\big)\big|
    \leq \sum_{i\in\bb Z_+^N}\binom{\ell}{i}\binom{m}{i}i!\gamma^i\big|\rho\big((a^*)^{\ell-i}a^{m-i}\big)\big|.
  \end{equation}
  On the other hand, by the Kadison-Schwarz inequality and \eqref{mcK}, for $i\leq\ell,m$
  \begin{equation*}
    \begin{split}
     &\big|\rho\big((a^*)^{\ell-i}a^{m-i}\big)\big|^2
    \leq \rho\big((a^*)^{\ell-i}a^{m-i}(a^*)^{m-i}a^{\ell-i}\big)\\
    &=\rho\Big(\prod_{k=1}^Nn_k(n_k-1)\cdots(n_k-\ell_k+i_k+1)\\
    &\qquad\times(n_k-\ell_k+i_k+1)(n_k-\ell_k+i_k+2)\cdots(n_k-\ell_k+m_k)\Big)\\
    &\leq \rho\Big((n+\ell+m-2i)^{\ell+m-2i}\Big)\\
    &= \rho\Big(\prod_{k\colon\ell_k+m_k-2i_k>0}(n_k+\ell_k+m_k-2i_k)^{\ell_k+m_k-2i_k}\Big)\\
    &\leq \rho\Big(\prod_{k\colon\ell_k+m_k-2i_k>0}2^{\ell_k+m_k-2i_k-1}\big(n_k ^{\ell_k+m_k-2i_k}+(\ell_k+m_k-2i_k)^{\ell_k+m_k-2i_k}\big)\Big)\\
    &\leq2^{|\ell+m-2 i|-|\{k\colon\ell_k+m_k-2i_k>0\}|}\sum_{A\subset \{k\colon\ell_k+m_k-2i_k>0\}}\prod_{\alpha\in A}(\ell_\alpha+m_\alpha-2i_\alpha)^{\ell_\alpha+m_\alpha-2i_\alpha}\\
    &\phantom{\leq2^{|\ell+m-2 i|-|\{k\colon\ell_k+m_k-2i_k>0\}|}}\qquad\times\prod_{\beta\not\in A}(\ell_\beta+m_\beta-2i_\beta)!K^{\ell_\beta+m_\beta-2i_\beta}.
    \end{split}
  \end{equation*}
  By using $k^k\leq k!e^k$ and assuming, without loss of generality, that $K\geq e$, we deduce
  \begin{equation*}
    \big|\rho\big((a^*)^{\ell-i}a^{m-i}\big)\big|
    \leq (2K)^{|\ell+m-2i|/2}\sqrt{(\ell+m-2i)!}.
  \end{equation*}
  Plugging this bound into \eqref{eq:bdf2} we get
  \begin{equation*}
    \begin{split}
        \big|\rho\big(g_{\ell,m}\big)\big|
        &\leq\sqrt{\ell!m!} \sum_{i\in\bb Z_+^N}\sqrt{\binom{\ell}{i}\binom{m}{i}\binom{\ell+m-2i}{\ell-i}}\gamma^i(2K)^{|\ell+m-2i|/2}\\
     &\leq\sqrt{\ell!m!\gamma^\ell\delta^m}2^{|\ell+m|/2}\sum_{i\leq\ell,m}\Big(\frac{4K}{\gamma}\Big)^{(\ell+m-2i)/2}
    \end{split}
  \end{equation*}
  where we used $\binom{k}{h}\leq2^{|k|}$ and $\delta\geq\gamma$.
  Above $4K/\gamma$ stands for $4K(\gamma_1^{-1},\dots, \gamma_N^{-1})\in(0,\infty)^N$.
  Without loss of generality we now assume $4K>e^{-\omega_k/2}/\varsh(\omega_k/2)$, which implies $\gamma_k/(4K)\leq1/2$, $k=1,\dots,N$.
  Hence,
  \begin{equation*}
        \big|\rho\big(g_{\ell,m}\big)\big|
     \leq\sqrt{\ell!m!\gamma^\ell\delta^m}\Big(\frac{16K}{\gamma}\Big) ^{(\ell+m)/2}.
   \end{equation*}
   The  bound \eqref{eq:bdf1} thus holds with $C_k=4\sqrt{K/\gamma_k}$, $k=1,\dots,N$, and $T_K$ such that $e^{-2\Lambda T_K}\max_{k=1,\dots,N}C_k\leq1/2$.

   To achieve the bound \eqref{label4}, we assume without loss of generality $t\geq T_K$ and observe that $\sigma(g)=\langle f_{0,0},g\rangle_\sigma$.
   In view of \eqref{eq:dichi} and \eqref{cla}, by Cauchy-Schwarz inequality,
   \begin{equation*}
     D_{\chi^2}(\rho P_t|\sigma)
     \leq \sum_{|j|>0}e^{-2\lambda_j t}\sum_{m\leq j}\big|\rho\big(f_{j,m}\big)\big|^2
     \leq \prod_{k=1}^N\sum_{j_k=0}^\infty(j_k+1)\big(e^{-4\varsh(\omega_k/2)t}C_k^{2}\big)^{j_k}-1,
   \end{equation*}
   where we used \eqref{spl} and \eqref{eq:bdf}.
   We then get
   \begin{equation*}
     D_{\chi^2}(\rho P_t|\sigma)
     \leq \prod_{k=1}^N\Big(1+A_ke^{-2\Lambda t}\Big)-1
   \end{equation*}
   where
   \begin{equation*}
     A_k:=\sum_{j=1}^\infty e^{-2\Lambda T_K(j-1)}(j+1)C_k^{2j}
     \leq A
   \end{equation*}
   for a suitable constant $A$ depending only on $\Lambda$ and $K$.

   It remains to prove the equality in \eqref{cla}.
   To this end, for $r\in \bb N$ let $\chi_r$ be the orthogonal projector on the span of $\big\{ |\ell\rangle, \: |\ell |\le r\big\}$ and introduce the (finite rank) density matrix
  \begin{equation*}
      \rho_r := \frac 1{Z_r} \chi_r \rho \chi_r
  \end{equation*}
  where $Z_r$ is the appropriate normalization constant. Denoting
  still by $\rho_r$ the state on $\mc K$ with this density matrix we
  readily deduce that the sequence $\{\rho_r\}$ converges weakly* to
  $\rho$. Letting $\varrho_{r,\sigma}$ be the finite rank operator
  given by
  \begin{equation*}
    \varrho_{r,\sigma}
    =\prod_{k=1}^N(1-e^{-\omega_k})\cdot\sum_{|\ell|\leq r}e^{\sum_{k=1}^N\omega_k\ell_k}|\ell\rangle\langle\ell|\rho_r,
  \end{equation*}
   for each $g\in\mc K$ and  $r\in
  \bb N$ we have 
  \begin{equation*}
    \rho_r(g) = \big\langle \varrho_{r,\sigma}, g  \big\rangle_\sigma.
  \end{equation*}
  Hence
  \begin{equation}
    \label{prelim}
    \big(\rho P_t\big)(g) =\rho\big(P_t g\big)=\lim_{r\to\infty}  \rho_r\big(P_t g\big)
    =\lim_{r\to\infty}\big\langle\varrho_{r,\sigma}, P_t g  \big\rangle_\sigma
    =\lim_{r\to\infty}\big\langle\varrho_{r,\sigma}, P_t^{(2)} g
    \big\rangle_\sigma. 
  \end{equation}
  On the other hand, by the spectral decomposition in
  Theorem~\ref{thm:specdec}, for each $g\in \mc K \subset L_2(\sigma)$
  \begin{equation}
    \label{toplug}
    P_t^{(2)} g = \sum_{j\in \bb Z_+^N} \sum_{m\le j}
    e^{-\lambda_j t} \: \big\langle f_{j,m}, g \big\rangle_\sigma \: f_{j,m}.
  \end{equation}
  In view of \eqref{mcK}, the definition of $\rho_r$, and the arguments used to deduce \eqref{eq:bdf}, letting $C\in(0,\infty)^N$ be as in \eqref{eq:bdf}, we have   
  \begin{equation*}
    \sup_{r\in \bb N} \big| \rho_r\big( f_{j,m} \big) \big| \le C^j.
  \end{equation*}
  Plug \eqref{toplug} on the right hand side of \eqref{prelim}.
  In view of the previous bound, the expression \eqref{spl} for the
  eigenvalues $\lambda_j$, and dominated convegence, we can 
  exhange the limit as $r\to \infty$ with the summation over $j\in \bb
  Z_+^N$ and deduce the claim \eqref{cla}.
\end{proof}

\subsection*{Commutative case: restriction to diagonal matrices.}

Consider the continuos time Markov chain on state space $\bb Z_+^N$
whose pre-generator $L_\mathrm{cl}$ acts on compactly supported functions
$f\colon \bb Z_+^N \to \bb R$ as
\begin{equation*}
  \big(L_\mathrm{cl} f\big) (\ell) = \sum_{k=1}^N 
  \Big\{ c_k^+(\ell_k)  \big[ f(\ell +e_k) - f(\ell) \big]
  +c_k^-(\ell_k)  \big[ f(\ell - e_k) - f(\ell) \big] \Big\} 
\end{equation*}
where we recall that $e_k$ has coordinates $(e_k)_{k'} =\delta_{k,k'}$ and 
the \emph{jump rates} are given by 
\begin{equation}\label{eq:jrc}
  c_k^+(l) = 2 \, (l +1) \, e^{-\omega_k/2},
  \qquad\qquad
  c_k^-(l) = 2\, l \,e^{\omega_k/2}.
\end{equation}
The chain generated by $L_\mathrm{cl}$ corresponds to $N$ independent
birth and death processes with birth and death  rate
given by $c_k^+$, $c_k^-$, respectively. It is simple to check that
\begin{equation*}
  L^{(0)} \sum_{\ell\in\bb Z_+^N} f(\ell) |\ell\rangle\langle\ell|
  =\sum_{\ell\in\bb Z_+^N} \big( L_\mathrm{cl} f\big)
  (\ell) \, |\ell\rangle\langle\ell|.
\end{equation*}
In particular, if the state $\rho$ has the form
$\rho= \sum_{\ell\in\bb Z_+^N}\mu(\ell)\,
|\ell\rangle\langle\ell|$ for some probability $\mu$ on  $\bb Z_+^N$,
then $\rho P_t =  \sum_{\ell\in\bb Z_+^N}
\big(\mu P_t^\mathrm{cl}\big) (\ell)\,  |\ell\rangle\langle\ell|$ where  
$(P_t^\mathrm{cl})_{t\geq0}$ is the Markov semigroup on $c_0(\bb Z_+^N)$ generated by $L_\mathrm{cl}$.
Cle\-ar\-ly, the generator $L_\mathrm{cl}$ is reversible with respect to the
probability $\pi$ on $\bb Z_+^N$ given by $\pi(\ell) = Z^{-1} e^{-\sum_k \omega_k \ell_k}$,
where $Z$ is the appropriate normalization. 
Theorem~\ref{thm:specdec} thus provides, as a corollary, the spectral
decomposition of $L_\mathrm{cl}^{(2)}$.
Here we have denoted by $L_\mathrm{cl}^{(2)}$  the generator of the semigroup $(P_t^{\mathrm{cl},(2)})_{t\geq0}$, the extension to $L^2(\bb Z_+^N,\pi)$ of $(P_t^\mathrm{cl})_{t\geq0}$.
More precisely, recalling
\eqref{spl}, 
$\mathrm{spec}(-L_\mathrm{cl}^{(2)})
=\{\lambda^\mathrm{cl}_\ell\}_{\ell\in\bb Z_+^N}$
with eigenvalues $\lambda_\ell^\mathrm{cl}$ given by
\begin{equation*}
  \lambda^\mathrm{cl}_\ell =2 \lambda_\ell
  =4 \sum_{k=1}^N  
   \varsh\big(\tfrac {\omega_k}2\big) \ell_k.
\end{equation*}
The corresponding orthonormal basis $\{f_\ell^\mathrm{cl}\}$ of $L^2(\bb Z_+^N,\pi)$ of eigenvectors is given by
\begin{equation*}
  f_\ell^\mathrm{cl}(m) = \frac 1{\gamma^{\ell/2}\delta^{\ell/2}}
  \sum_{j\in \bb Z_+^N} \binom {\ell}{j}\binom {m}{\ell -j} (-1)^{|j|}
  \gamma^j  
  \,,
\end{equation*}
where $\gamma$ and $\delta$ are given by \eqref{al=}.
Indeed, by \eqref{plm} and \eqref{aul},
$$
f_{2\ell,\ell} = \sum_m f_\ell^\mathrm{cl}(m) \,|m\rangle\langle m|.
$$

Furthermore, Theorem \ref{thm-bose} implies the following statement regarding the exponential ergodicity
of $(P_t^{\mathrm{cl}})_{t\geq0}$.
Fix $K>0$ and denote by $\mc P_K$ the set of probabilities $\mu$ on $\bb Z_+^N$
satisfying
\begin{equation}
  \label{eq:skc}
  \sum_{m\in\bb Z_+^N}\mu(m) m^{\ell}
  \leq
  K^{|\ell|}
  \ell!
  \,.  
\end{equation}
Then, for each $K>0$ there exists $A=A(K,\Lambda)$ such that for each $N\in\bb N$ and $t\geq0$
$$
4 \sup_{\mu\in\mc P_K}
    \mathrm{d}_\mathrm{TV} (\mu P^{\mathrm{cl}}_t,\pi)^2
     \le 
\big(
 1+Ae^{-4\Lambda t}
\big)^N -1
\,,
$$
where $\mathrm{d}_\mathrm{TV}$ denotes the total variation distance.

The classical example of a birth and death chain reversible with respect to $\pi$ has jump rates $c_k^+(l)=e^{-\omega_k}$, $c_k^-(l)=\ind(l_k>0)$.
It provides a popular example for which the modified logarithmic Sobolev inequality fails.
In contrast, as follows from \cite[Thm.8.5]{CM1}, this inequality holds for the jump rates specified in \eqref{eq:jrc}.
The value of the modified logarithmic Sobolev constant is $4\Lambda$.
By applying the general bound for one-dimensional chain in \cite{Mi}, it is also possible to prove directly a logarithmic Sobolev inequality uniform in $N$ but the value of the corresponding constant is not sharp.
To our knowledge, the convergence in total variation for initial states satisfying the bound \eqref{eq:skc} has not been previously analyzed in the literature. 

\appendix

\section{Gap of the Lindblad generator for simple Lie algebras}\label{App1}

In the case of classical Lie algebras, i.e. for types $A_r$, $B_r$, $C_r$ and $D_r$,
we rely on specific representations of the root systems, see e.g. \cite{Hum}.
For the exceptional Lie algebras $E_6$, $E_7$, $E_8$, $F_4$ and $G_2$, 
we do not rely on a specific representation of the root system but only 
on the computation of the inverse of the Cartan matrix.

\subsection*{Case $A_r=\mf{sl}_{n}$, $n=r+1\geq2$}

The root system of type $A_r$ is
$$
\Phi
=
\big\{
e_i-e_j
\big\}_{1\leq i\neq j\leq n}
$$
where $e_1,\dots,e_{n}$ is the canonical basis of the Euclidean space $\bb R^{n}$.
In this case $E=(e_1+\dots+e_{n})^{\perp}\subset\bb R^{n}$.
We fix the base $\Delta=\{\alpha_1,\dots,\alpha_r\}$,
where 
$\alpha_i=e_i-e_{i+1}$.
With this choice of the base, the subset of positive roots is
$\Phi_+
=
\big\{
e_i-e_j
\big\}_{1\leq i<j\leq n}$,
the highest root is $\theta=e_1-e_n$,
and the sum of positive roots is
$$
2\delta
=
\sum_{\alpha\in\Phi_+}\alpha
=
\sum_{i=1}^n(n+1-2i)e_i
\,.
$$
The root lattice is
$P=\Span_{\bb Z}\{\Phi\}=\bb Z^n\cap E$,
while its positive cone is
$$
P^+
=
\Span_{\bb Z_+}\{\Phi_+\}
=
\bigg\{
\sum_{i=1}^nk_ie_i\in \bb Z^n\cap E\,\colon\,
\sum_{j=1}^i k_j\geq0,\, i=1,\dots,n
\bigg\}
\,.
$$

The system of fundamental weights 
associated to the base $\Delta$ is $\Pi=\{\wp_1,\dots,\wp_r\}$,
where
$$
\wp_i
=
\frac{n-i}n\sum_{j=1}^ie_j-\frac{i}{n}\sum_{j=i+1}^ne_j
\,,\quad i=1,\dots,r
\,.
$$
The corresponding weight lattice is
$$
\Lambda
=
\Span_{\bb Z}\{\Pi\}
=
\bigg\{
\sum_{i=1}^n(k_i-\bar k)e_i
\,\colon\,
k_1,\dots,k_n\in\bb Z
\bigg\}
\,,
$$
where
$\bar k=\frac1n\sum_{i=1}^nk_i$.
In other words, $\Lambda$ is the orthogonal projection of $\bb Z^n$ onto $E$.
The positive cone of dominant integral weights is
$$
\Lambda^+
=
\Span_{\bb Z_+}\{\Pi\}
=
\bigg\{
\sum_{i=1}^n(k_i-\bar k)e_i\in\Lambda
\,\colon\,
k_1\geq k_2\geq \dots\geq k_n
\bigg\}
\,.
$$
It follows that
$$
P^+\cap\Lambda^+
=
\bigg\{
\sum_{i=1}^nk_ie_i\in\bb Z^n
\,\colon\,
k_1\geq k_2\geq \dots\geq k_{n-1}\geq0,\,
\sum_{i=1}^{n}k_i=0
\bigg\}
\,.
$$

Observe that, if $\mu=\sum_{i=1}^nk_ie_i\in P^+\cap\Lambda^+$,
then 
$$
(\mu|\mu+2\delta)
=
\sum_{i=1}^nk_i(k_i+n+1-2i)
=
\sum_{i=1}^nk_i(k_i+2(n-i))
\,,
$$
where, for the second equality, we used $\sum_{i=1}^nk_i=0$.
Now each $i$-th summand in the right hand side is positive and increasing in $k_i$.
Hence, the minimum in \eqref{eq:g0}
is achieved for $\mu=e_1-e_n=\theta$,
which corresponds to the adjoint representation $V_\theta=\ad \mf g$.
As a consequence $g_0=1$, as claimed in Theorem \ref{t:mrl}.

\subsection*{Case $B_r=\mf{so}_{2r+1}$, $r\geq2$}

The root system of type $B_r$ is
$$
\Phi
=
\big\{
\pm e_i\pm e_j
\big\}_{1\leq i<j\leq r}
\sqcup
\big\{
\pm e_i
\big\}_{1\leq i\leq r}
$$
where $e_1,\dots,e_{r}$ is the canonical basis of the Euclidean space $E=\bb R^{r}$.
We fix the base $\Delta=\{\alpha_1,\dots,\alpha_r\}$,
where 
$\alpha_i=e_i-e_{i+1}$ for $i=1,\dots,r-1$ and $\alpha_r=e_r$.
With this choice of the base, the subset of positive roots is
$$
\Phi_+
=
\big\{
e_i\pm e_j
\big\}_{1\leq i<j\leq r}
\sqcup
\big\{
e_i
\big\}_{1\leq i\leq r}
\,,
$$
the highest root is $\theta=e_1+e_2$,
and the sum of positive roots is
$$
2\delta
=
\sum_{\alpha\in\Phi_+}\alpha
=
\sum_{i=1}^r (2r+1-2i) e_i
\,.
$$
The root lattice is
$P=\Span_{\bb Z}\{\Phi\}=\bb Z^r$,
while its positive cone is
$$
P^+
=
\Span_{\bb Z_+}\{\Phi_+\}
=
\bigg\{
\sum_{i=1}^rk_ie_i\in \bb Z^r\,\colon\,
\sum_{j=1}^i k_j\geq0,\, i=1,\dots,r
\bigg\}
\,.
$$

The system of fundamental weights 
associated to the base $\Delta$ is $\Pi=\{\wp_1,\dots,\wp_r\}$,
where
$$
\wp_i
=
\sum_{j=1}^ie_j\,,\quad i=1,\dots,r-1\,,\qquad
\wp_r=\frac12\sum_{j=1}^re_j
\,.
$$
The corresponding weight lattice is
$\Lambda
=
\Span_{\bb Z}\{\Pi\}
=
\Lambda_{\bar 0}\sqcup\Lambda_{\bar 1}$,
where
$$
\Lambda_{\bar 0}
=
\bigg\{
\sum_{i=1}^rk_ie_i
\,\colon\,
k_1,\dots,k_r\in\bb Z
\bigg\}
\,\text{ and }\,
\Lambda_{\bar 1}
=
\bigg\{
\sum_{i=1}^rk_ie_i
\,\colon\,
k_1,\dots,k_r\in\frac12+\bb Z
\bigg\}
\,.
$$
The positive cone of dominant integral weights is
$\Lambda^+
=
\Span_{\bb Z_+}\{\Pi\}
=
\Lambda_{\bar 0}^+\sqcup\Lambda_{\bar 1}^+$,
where
$$
\Lambda_{a}^+
=
\bigg\{
\sum_{i=1}^rk_ie_i\in\Lambda_a
\,\colon\,
k_1\geq k_2\geq \dots\geq k_r\geq0
\bigg\}
\,,\quad a\in\{\bar 0,\bar 1\}
\,.
$$
It follows that
$P^+\cap\Lambda^+
=
\Lambda_{\bar 0}^+$.

Observe that, if $\mu=\sum_{i=1}^rk_ie_i\in\Lambda_{\bar 0}^+$,
then 
$$
(\mu|\mu+2\delta)
=
\sum_{i=1}^rk_i(k_i+2r+1-2i)
\,.
$$
Now each $i$-th summand in the right hand side is positive and increasing in $k_i$.
Hence, the minimum in \eqref{eq:g0}
is achieved for $\mu=e_1=\wp_1$,
where it has value
$(\wp_1|\wp_1+2\delta)=2r$.
On the other hand, the highest root is $\theta=e_1+e_2=\wp_2$,
and
$(\theta|\theta+2\delta)
=
2(2r-1)$.
As a consequence $g_0=\frac{r}{2r-1}<1$, as claimed in Theorem \ref{t:mrl}.

We conclude the study of this case
by exhibiting a representation of $\mf g$ for which 
$\gap(-L)=g_0$ and another representation for which $\gap(-L)=1$.
For the defining representation $V=V_{\wp_1}$ we have
$\End V\simeq V_0\oplus V_{\theta}\oplus V_{2\wp_1}$,
so that, in this case, $V_{\wp_1}$ does not appear and $\gap(-L)=1$.
On the other hand, for $V=V_{\wp_r}$
we have
$\End V\simeq V_0\oplus\Big(\bigoplus_{i=1}^{r-1}V_{\wp_i}\Big)\oplus V_{2\wp_r}$.
Since, in this case, $V_{\wp_1}$ does appear in the decomposition,
we have $\gap(-L)=g_0$.

\subsection*{Case $C_r=\mf{sp}_{2r}$, $r\geq3$}

The root system of type $C_r$ is
$$
\Phi
=
\big\{
\pm e_i\pm e_j
\big\}_{1\leq i<j\leq r}
\sqcup
\big\{
\pm 2e_i
\big\}_{1\leq i\leq r}
$$
where $e_1,\dots,e_{r}$ is the canonical basis of the Euclidean space $E=\bb R^{r}$.
We fix the base $\Delta=\{\alpha_1,\dots,\alpha_r\}$,
where 
$\alpha_i=e_i-e_{i+1}$ for $i=1,\dots,r-1$ and $\alpha_r=2e_r$.
With this choice of the base, the subset of positive roots is
$$
\Phi_+
=
\big\{
e_i\pm e_j
\big\}_{1\leq i<j\leq r}
\sqcup
\big\{
2e_i
\big\}_{1\leq i\leq r}
\,,
$$
the highest root is $\theta=2e_1$,
and the sum of positive roots is
$$
2\delta
=
\sum_{\alpha\in\Phi_+}\alpha
=
2\sum_{i=1}^r (r+1-i) e_i
\,.
$$
The root lattice is
$$
P=\Span_{\bb Z}\{\Phi\}=
\Big\{\sum_{i=1}^rk_ie_i\in\bb Z^r
\,\colon\,
\sum_{i=1}^rk_i\in2\bb Z
\Big\}
\,,
$$
while its positive cone is
$$
P^+
=
\Span_{\bb Z_+}\{\Phi_+\}
=
\bigg\{
\sum_{i=1}^rk_ie_i\in P\,\colon\,
\sum_{j=1}^i k_j\geq0,\, i=1,\dots,r
\bigg\}
\,.
$$

The system of fundamental weights 
associated to the base $\Delta$ is $\Pi=\{\wp_1,\dots,\wp_r\}$,
where
$$
\wp_i
=
\sum_{j=1}^ie_j\,,\quad i=1,\dots,r
\,.
$$
The corresponding weight lattice is
$\Lambda
=
\Span_{\bb Z}\{\Pi\}
=\bb Z^r$,
while 
the positive cone of dominant integral weights is
$$
\Lambda^+
=
\Span_{\bb Z_+}\{\Pi\}
=
\bigg\{
\sum_{i=1}^rk_ie_i\in\bb Z^r
\,\colon\,
k_1\geq k_2\geq \dots\geq k_r\geq0
\bigg\}
\,.
$$
It follows that
$$
P^+\cap\Lambda^+
=
\bigg\{
\sum_{i=1}^rk_ie_i\in\bb Z^r
\,\colon\,
k_1\geq k_2\geq \dots\geq k_r\geq0,\,
\sum_{i=1}^r k_i\in 2\bb Z
\bigg\}
\,.
$$

Observe that, if $\mu=\sum_{i=1}^rk_ie_i\in P^+\cap\Lambda^+$,
then 
$$
(\mu|\mu+2\delta)
=
\sum_{i=1}^rk_i(k_i+2(r+1-i))
\,.
$$
Now, each $i$-th summand in the right hand side is positive and increasing in $k_i$.
Hence, the minimum in \eqref{eq:g0}
is achieved either for the highest root $\theta=2e_1$, corresponding to the adjoint representation,
or for $\wp_2=e_1+e_2$.
By computing the two values, we get
$(\theta|\theta+2\delta)=4(r+1)$
and $(\wp_2|\wp_2+2\delta)=4r$.
As a consequence $g_0=\frac{r}{r+1}<1$, as claimed in Theorem \ref{t:mrl}.

Also in this case, we conclude
by exhibiting a representation of $\mf g$ for which 
$\gap(-L)=g_0$ and another representation for which $\gap(-L)=1$.
For the 
defining representation $V=V_{\wp_1}$ we have
$\End V\simeq V_0\oplus V_{\theta}\oplus V_{\wp_2}$,
so that, in this case, $V_{\wp_2}$ does appear in the decomposition 
and $\gap(-L)=g_0$.
On the other hand, for $V=V_{\wp_r}$
we have
$\End V\simeq V_0\oplus\Big(\bigoplus_{i=1}^{r}V_{2\wp_i}\Big)$.
Since, in this case, $V_{\wp_2}$ does not appear,
we have $\gap(-L)=1$.

\subsection*{Case $D_r=\mf{so}_{2r}$, $r\geq4$}

The root system of type $D_r$ is
$$
\Phi
=
\big\{
\pm e_i\pm e_j
\big\}_{1\leq i<j\leq r}
$$
where $e_1,\dots,e_{r}$ is the canonical basis of the Euclidean space $E=\bb R^{r}$.
We fix the base $\Delta=\{\alpha_1,\dots,\alpha_r\}$,
where 
$\alpha_i=e_i-e_{i+1}$ for $i=1,\dots,r-1$ and $\alpha_r=e_{r-1}+e_r$.
With this choice of the base, the subset of positive roots is
$$
\Phi_+
=
\big\{
e_i\pm e_j
\big\}_{1\leq i<j\leq r}
\,,
$$
the highest root is $\theta=e_1+e_2$,
and the sum of positive roots is
$$
2\delta
=
\sum_{\alpha\in\Phi_+}\alpha
=
2\sum_{i=1}^r (r-i) e_i
\,.
$$
The root lattice is
$$
P
=
\Span_{\bb Z}\{\Phi\}
=
\bigg\{
\sum_{i=1}^rk_ie_i\in \bb Z^r\,\colon\,
\sum_{i=1}^r k_i\in2\bb Z
\bigg\}
\,,
$$
while its positive cone is
$$
P^+
=
\Span_{\bb Z_+}\{\Phi_+\}
=
\bigg\{
\sum_{i=1}^rk_ie_i\in P\,\colon\,
\sum_{j=1}^i k_j\geq0,\, i=1,\dots,r-1
,\,
|k_r|\leq \sum_{j=1}^{r-1} k_j
\bigg\}
\,.
$$

The system of fundamental weights 
associated to the base $\Delta$ is $\Pi=\{\wp_1,\dots,\wp_r\}$,
where
$$
\wp_i
=
\sum_{j=1}^ie_j\,,\quad i=1,\dots,r-2
\,,\qquad
\wp_{r-1}=\frac12\Big(\sum_{i=1}^{r-1}e_i-e_r\Big)
\,,\quad
\wp_{r-1}=\frac12\sum_{i=1}^{r}e_i
\,.
$$
The corresponding weight lattice is
$\Lambda
=
\Span_{\bb Z}\{\Pi\}
=
\Lambda_{\bar 0}\sqcup\Lambda_{\bar 1}$,
where
$$
\Lambda_{\bar 0}
=
\bigg\{
\sum_{i=1}^rk_ie_i
\,\colon\,
k_1,\dots,k_r\in\bb Z
\bigg\}
\,\text{ and }\,
\Lambda_{\bar 1}
=
\bigg\{
\sum_{i=1}^rk_ie_i
\,\colon\,
k_1,\dots,k_r\in\frac12+\bb Z
\bigg\}
\,.
$$
The positive cone of dominant integral weights is
$\Lambda^+
=
\Span_{\bb Z_+}\{\Pi\}
=
\Lambda_{\bar 0}^+\sqcup\Lambda_{\bar 1}^+$,
where
$$
\Lambda_{a}^+
=
\bigg\{
\sum_{i=1}^rk_ie_i\in\Lambda_a
\,\colon\,
k_1\geq k_2\geq \dots\geq k_{r-1}\geq|k_r|
\bigg\}
\,,\quad a\in\{\bar 0,\bar 1\}
\,.
$$
It follows that
$$
P^+\cap\Lambda^+
=
\bigg\{
\sum_{i=1}^rk_ie_i\in\bb Z^r
\,\colon\,
k_1\geq k_2\geq \dots\geq k_{r-1}\geq|k_r|,\,
\sum_{i=1}^rk_i\in 2\bb Z
\bigg\}
$$

Observe that, if $\mu=\sum_{i=1}^rk_ie_i\in P^+\cap\Lambda^+$,
then 
$$
(\mu|\mu+2\delta)
=
\sum_{i=1}^rk_i(k_i+2(r-i))
\,.
$$
Now each $i$-th summand in the right hand side is positive and increasing in $k_i$.
Hence, the minimum in \eqref{eq:g0}
is achieved 
either for the highest root $\theta=e_1+e_2=\wp_2$, corresponding to the adjoint representation,
or for $2e_1=2\wp_1$.
By computing the two values, we get
$(\theta|\theta+2\delta)=4(r-1)
<
4r=(2\wp_1|2\wp_1+2\delta)$.
As a consequence $g_0=1$, as claimed in Theorem \ref{t:mrl}.

\subsection*{Case $E_6$}

The Dynkin diagram of $E_6$ is
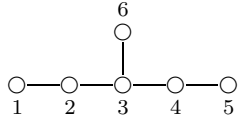
\begin{center}
\begin{picture}(100,46)
\put(65,30){\circle{6}}
\put(65,26){\line(0,-1){12}}
\multiput(25,10)(20,0){5}{\circle{6}}
\multiput(29,10)(20,0){4}{\line(1,0){12}}
\put(23,-1){\footnotesize{1}}
\put(43,-1){\footnotesize{2}}
\put(63,-1){\footnotesize{3}}
\put(83,-1){\footnotesize{4}}
\put(103,-1){\footnotesize{5}}
\put(63,36){\footnotesize{6}}
\end{picture}
\end{center}
The corresponding Cartan matrix thus reads:
$$
A
=
\left(\begin{array}{cccccc}
2&-1&0&0&0&0 \\
-1&2&-1&0&0&0 \\
0&-1&2&-1&0&-1 \\
0&0&-1&2&-1&0 \\
0&0&0&-1&2&0 \\
0&0&-1&0&0&2
\end{array}\right)
$$
and its inverse is
$$
A^{-1}
=
\frac13
\left(\begin{array}{cccccc}
4&5&6&4&2&3 \\
5&10&12&8&4&6 \\
6&12&18&12&6&9 \\
4&8&12&10&5&6 \\
2&4&6&5&4&3 \\
3&6&9&6&3&6
\end{array}\right)
\,.
$$
We denote by $\alpha_1,\dots,\alpha_6$ the simple roots 
and by $\wp_1,\dots,\wp_6$ the associated fundamental weights.
Up to an overall normalization constant,
the Cartan matrix $A$ is the Gram matrix 
for the basis of simple roots $\alpha_1,\dots,\alpha_6$,
while the inverse matrix $A^{-1}$ is the Gram matrix for the basis 
of fundamental weight $\wp_1,\dots,\wp_6$.
In particular, $(\alpha_1,\dots,\alpha_6)=(\wp_1,\dots,\wp_6)A$.
By a direct computation, 
the maximal root is
$$
\theta
=
\alpha_1+2\alpha_2+3\alpha_3+2\alpha_4+\alpha_5+2\alpha_6
=
\wp_6
\,,
$$
and the sum of positive  roots is
\begin{equation}\label{deltaE6}
2\delta
=
16\alpha_1+30\alpha_2+42\alpha_3+30\alpha_4+16\alpha_5+22\alpha_6
\,.
\end{equation}
By definition, $P^+=\sum_{i=1}^6\bb Z_+\alpha_i$,
while $\Lambda^+=\sum_{i=1}^6\bb Z_+\wp_i$.
Hence,
$$
P^+\cap\Lambda^+
=
\Big\{
\sum_{i=1}^6n_i\wp_i
\,:\,
n\in\bb Z_+^6\,,\,\,
A^{-1}n\in\bb Z_+^6
\Big\}
\,,
$$
where $n^T=(n_1,\dots,n_6)$.
By the explicit expression of the matrix $A^{-1}$,
the condition $A^{-1}n\in\bb Z_+^6$
is equivalent to the equation
\begin{equation}\label{capE6}
n_1-n_2+n_4-n_5\equiv 0 \mod 3
\,.
\end{equation}

For $\mu=\sum_{i=1}^6n_i\wp_i$, we have
$$
(\mu|\mu+2\delta)
=
n^T A^{-1} n+R^Tn
\,,
$$
where $R^T=(16,30,42,30,16,22)$ (cf. \eqref{deltaE6}).

Since both the matrix $A^{-1}$ and the vector $R$ have positive entries,
a direct comparison shows that the minimum of $(\mu|\mu+2\delta)$
for $\mu=\sum_in_i\wp_i$ with $n_i\in\bb Z_+$ satisfying \eqref{capE6}
is achieved when $n_i=\delta_{6,i}$,
i.e. when $\mu=\wp_6=\theta$.
Hence, in this case $\gap(-L)=g_0=1$, as claimed in Theorem \ref{t:mrl}.

\subsection*{Case $E_7$}

The Dynkin diagram of $E_7$ is
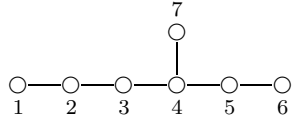
\begin{center}
\begin{picture}(100,46)
\put(65,30){\circle{6}}
\put(65,26){\line(0,-1){12}}
\multiput(5,10)(20,0){6}{\circle{6}}
\multiput(9,10)(20,0){5}{\line(1,0){12}}
\put(3,-1){\footnotesize{1}}
\put(23,-1){\footnotesize{2}}
\put(43,-1){\footnotesize{3}}
\put(63,-1){\footnotesize{4}}
\put(83,-1){\footnotesize{5}}
\put(103,-1){\footnotesize{6}}
\put(63,36){\footnotesize{7}}
\end{picture}
\end{center}
The corresponding Cartan matrix thus reads:
$$
A
=
\left(\begin{array}{ccccccc}
2&-1&0&0&0&0&0 \\
-1&2&-1&0&0&0&0 \\
0&-1&2&-1&0&0&0 \\
0&0&-1&2&-1&0&-1 \\
0&0&0&-1&2&-1&0 \\
0&0&0&0&-1&2&0 \\
0&0&0&-1&0&0&2
\end{array}\right)
$$
and its inverse is
$$
A^{-1}
=
\frac12
\left(\begin{array}{ccccccc}
3&4&5&6&4&2&3 \\
4&8&10&12&8&4&6 \\
5&10&15&18&12&6&9 \\
6&12&18&24&16&8&12 \\
4&8&12&16&12&6&8 \\
2&4&6&8&6&4&4 \\
3&6&9&12&8&4&7
\end{array}\right)
\,.
$$
As before, we denote by $\alpha_1,\dots,\alpha_7$ the simple roots 
and by $\wp_1,\dots,\wp_7$ the associated fundamental weights.
By a direct computation, 
the maximal root is
$$
\theta
=
\alpha_1+2\alpha_2+3\alpha_3+4\alpha_4+3\alpha_5+2\alpha_6+2\alpha_7
=
\wp_6
\,,
$$
and the sum of positive  roots is
\begin{equation}\label{deltaE7}
2\delta
=
27\alpha_1+52\alpha_2+75\alpha_3+96\alpha_4+66\alpha_5+34\alpha_6+49\alpha_7
\,.
\end{equation}
We have
$$
P^+\cap\Lambda^+
=
\Big\{
\sum_{i=1}^6n_i\wp_i
\,:\,
n\in\bb Z_+^7\,,\,\,
A^{-1}n\in\bb Z_+^7
\Big\}
\,,
$$
where $n^T=(n_1,\dots,n_7)$.
By the explicit expression of the matrix $A^{-1}$,
the condition $\sum_{j=1}^7 (A^{-1})_{ij}n_j\in\bb Z_+$
is equivalent to the equation
\begin{equation}\label{capE7}
n_1+n_3+n_7 \equiv 0 \mod 2
\,.
\end{equation}

For $\mu=\sum_{i=1}^7n_i\wp_i$, we have
$$
(\mu,\mu+2\delta)
=
n^T A^{-1} n+R^Tn
\,,
$$
where $R^T=(27,52,75,96,66,34,49)$ (cf. \eqref{deltaE7}).

Since both the matrix $A^{-1}$ and the vector $R$ have positive entries,
a direct comparison shows that
the minimum of $(\mu|\mu+2\delta)$
for $\mu=\sum_in_i\wp_i$ with $n_i\in\bb Z_+$ satisfying \eqref{capE7}
is achieved when $n_i=\delta_{6,i}$,
i.e. when $\mu=\wp_6=\theta$.
Hence, in this case $\gap(-L)=g_0=1$, as claimed in in Theorem \ref{t:mrl}..

\subsection*{Case $E_8$}

The Dynkin diagram of $E_8$ is
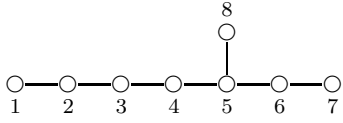
\begin{center}
\begin{picture}(100,46)
\put(65,30){\circle{6}}
\put(65,26){\line(0,-1){12}}
\multiput(-15,10)(20,0){7}{\circle{6}}
\multiput(-11,10)(20,0){6}{\line(1,0){12}}
\put(-17,-1){\footnotesize{1}}
\put(3,-1){\footnotesize{2}}
\put(23,-1){\footnotesize{3}}
\put(43,-1){\footnotesize{4}}
\put(63,-1){\footnotesize{5}}
\put(83,-1){\footnotesize{6}}
\put(103,-1){\footnotesize{7}}
\put(63,36){\footnotesize{8}}
\end{picture}
\end{center}
The corresponding Cartan matrix thus reads:
$$
A
=
\left(\begin{array}{cccccccc}
2&-1&0&0&0&0&0&0 \\
-1&2&-1&0&0&0&0&0 \\
0&-1&2&-1&0&0&0&0 \\
0&0&-1&2&-1&0&0&0 \\
0&0&0&-1&2&-1&0&-1 \\
0&0&0&0&-1&2&-1&0 \\
0&0&0&0&0&-1&2&0 \\
0&0&0&0&-1&0&0&2
\end{array}\right)
$$
and its inverse is
$$
A^{-1}
=
\left(\begin{array}{cccccccc}
2&3&4&5&6&4&2&3 \\
3&6&8&10&12&8&4&6 \\
4&8&12&15&18&12&6&9 \\
5&10&15&20&24&16&8&12 \\
6&12&18&24&30&20&10&15 \\
4&8&12&16&20&14&7&10 \\
2&4&6&8&10&7&4&5 \\
3&6&9&12&15&10&5&8
\end{array}\right)
\,.
$$
As before, we denote by $\alpha_1,\dots,\alpha_8$ the simple roots 
and by $\wp_1,\dots,\wp_8$ the associated fundamental weights.
By a direct computation, 
the maximal root is
$$
\theta
=
2\alpha_1+3\alpha_2+4\alpha_3+5\alpha_4+6\alpha_5+4\alpha_6+2\alpha_7+3\alpha_8
=
\wp_1
\,,
$$
and the sum of positive  roots is
\begin{equation}\label{deltaE8}
2\delta
=
58\alpha_1+114\alpha_2+168\alpha_3+220\alpha_4
+270\alpha_5+182\alpha_6+92\alpha_7+136\alpha_8
\,.
\end{equation}
In this case, since the matrix $A^{-1}$ has integral entries,
we have 
$$
P^+\cap\Lambda^+=\Lambda^+=
\Big\{
\sum_{i=1}^8n_i\wp_i
\,:\,
n_i\in\bb Z_+
\Big\}
\,.
$$

For $\mu=\sum_{i=1}^8n_i\wp_i$, we have
$$
(\mu|\mu+2\delta)
=
n^T A^{-1} n+R^Tn
\,,
$$
where $R^T=(58,114,168,220,270,182,92,136)$ (cf. \eqref{deltaE8}).

Since both the matrix $A^{-1}$ and the vector $R$ have positive entries,
a direct comparison shows that
the minimum of $(\mu|\mu+2\delta)$
for $\mu=\sum_in_i\wp_i$ with $n_i\in\bb Z_+$
is achieved when $n_i=\delta_{1,i}$,
i.e. when $\mu=\wp_1=\theta$.
Hence, in this case $\gap(-L)=g_0=1$, as claimed in in Theorem \ref{t:mrl}.

\subsection*{Case $F_4$}

The Dynkin diagram of $F_4$ is
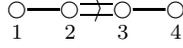
\begin{center}
\begin{picture}(100,26)
\multiput(25,10)(20,0){4}{\circle{6}}
\put(29,10){\line(1,0){12}}
\put(69,10){\line(1,0){12}}
\put(49,12){\line(1,0){12}}
\put(49,8){\line(1,0){12}}
\put(54,8){$\rangle$}
\put(23,-1){\footnotesize{1}}
\put(43,-1){\footnotesize{2}}
\put(63,-1){\footnotesize{3}}
\put(83,-1){\footnotesize{4}}
\end{picture}
\end{center}

The corresponding Cartan matrix thus reads:
$$
A
=
\left(\begin{array}{cccc}
2&-1&0&0 \\
-1&2&-1&0 \\
0&-2&2&-1 \\
0&0&-1&2
\end{array}\right)
\,.
$$
We denote by $\alpha_1,\dots,\alpha_4$ the simple roots 
and by $\wp_1,\dots,\wp_4$ the associated fundamental weights,
which are defined by 
$\frac{2(\alpha_1|\wp_j)}{(\alpha_i|\alpha_i)}=\delta_{i,j}$.
The inverse of the Cartan matrix is
$$
A^{-1}
=
\left(\begin{array}{cccc}
2&3&2&1 \\
3&6&4&2 \\
4&8&6&3 \\
2&4&3&2
\end{array}\right)
\,.
$$
By the definition of the fundamental weights,
$(\alpha_1,\dots,\alpha_4)=(\wp_1,\dots,\wp_4)A$
and 
$(\wp_1,\dots,\wp_4)=(\alpha_1,\dots,\alpha_4)A^{-1}$.

Assuming, without loss of generality,
that $||\alpha_1||^2=||\alpha_2||^2=2$ and $||\alpha_3||^2=||\alpha_4||^2=1$,
the Gram matrix for the fundamental weights $\wp_1,\dots,\wp_4$ is
$$
S
=
D
A^{-1}
=
\left(\begin{array}{cccc}
2&3&2&1 \\
3&6&4&2 \\
2&4&3&\frac32 \\
1&2&\frac32&1
\end{array}\right)
\quad
\text{where}
\quad
D
=
\left(\begin{array}{cccc}
1&0&0&0 \\
0&1&0&0 \\
0&0&\frac12&0 \\
0&0&0&\frac12
\end{array}\right)
\,.
$$

By a direct computation, 
the maximal root is
$$
\theta
=
2\alpha_1+3\alpha_2+4\alpha_3+2\alpha_4=\wp_1
\,,
$$
and the sum of positive  roots is
\begin{equation}\label{deltaF4}
2\delta
=
16\alpha_1+30\alpha_2+42\alpha_3+22\alpha_4
\,.
\end{equation}

By definition, $P^+=\sum_{i=1}^4\bb Z_+\alpha_i$,
while $\Lambda^+=\sum_{i=1}^4\bb Z_+\wp_i$
and, since $A^{-1}$ has integer positive entries, $\Lambda^+\subset P^+$.
Hence,
$P^+\cap\Lambda^+=\Lambda^+$.

For $\mu=\sum_{i=1}^4n_i\wp_i$, we have
$$
(\mu|\mu+2\delta)
=
n^T S n+R^TDn
\,,
$$
where $R^T=(16,30,42,22)$ (cf. \eqref{deltaF4}).

Since both the matrix $S$ and the vector $R$ have positive entries,
a direct comparison shows that the minimum of $(\mu|\mu+2\delta)$
is achieved when $n_i=\delta_{4,i}$,
i.e. when $\mu=\wp_4$.
The corresponding value is
$(\wp_4|\wp_4+2\delta)
=
12$.
On the other hand, we have 
$(\theta|\theta+2\delta)
=
18$.
As a consequence $g_0=\frac23$, as claimed in in Theorem \ref{t:mrl}..

We conclude
by exhibiting a representation of $F_4$ for which 
$\gap(-L)=\frac23$ and another representation for which $\gap(-L)=1$.
For the 
$26$-dimensional representation $V=V_{\wp_4}$ we have
$\End V\simeq V_0\oplus V_{\wp_4}\oplus V_{2\wp_4}\oplus V_{\wp_3}\oplus V_{\wp_1}$,
so that, in this case, $V_{\wp_4}$ does appear in the decomposition 
and $\gap(-L)=\frac23$.
On the other hand, for the adjoint representation $V=V_{\wp_1}$
we have
$\End V\simeq V_0\oplus V_{\wp_1}\oplus V_{2\wp_1}\oplus V_{\wp_2}\oplus V_{2\wp_4}$.
Since, in this case, $V_{\wp_4}$ does not appear,
we have $\gap(-L)=1$.

\subsection*{Case $G_2$}

The Dynkin diagram of $G_2$ is
\begin{center}
\begin{picture}(100,26)
\multiput(45,10)(20,0){2}{\circle{6}}
\put(49,12){\line(1,0){12}}
\put(49,10){\line(1,0){12}}
\put(49,8){\line(1,0){12}}
\put(54,8){$\rangle$}
\put(43,-1){\footnotesize{1}}
\put(63,-1){\footnotesize{2}}
\end{picture}
\end{center}
The corresponding Cartan matrix thus reads:
$$
A
=
\left(\begin{array}{cc}
2&-1 \\
-3&2
\end{array}\right)
\,.
$$
We denote by $\alpha_1,\alpha_2$ the simple roots 
and by $\wp_1,\wp_2$ the associated fundamental weights,
defined by 
$\frac{2(\alpha_1|\wp_j)}{(\alpha_i|\alpha_i)}=\delta_{i,j}$.
The inverse of the Cartan matrix is
$$
A^{-1}
=
\left(\begin{array}{cccc}
2&1 \\
3&2
\end{array}\right)
\,.
$$
Hence, 
$\alpha_1=2\wp_1-3\wp_2$, $\alpha_2=\wp_1+2\wp_2$
and 
$\wp_1=2\alpha_1+3\alpha_2$, $\wp_2=\alpha_1+2\alpha_2$.

By the definition of the Cartan matrix, we have $\frac{2(\alpha_1|\alpha_2)}{(\alpha_1|\alpha_1)}=-1$
and $\frac{2(\alpha_2|\alpha_1)}{(\alpha_2|\alpha_2)}=-3$, so that
$\frac{(\alpha_1|\alpha_1)}{(\alpha_2|\alpha_2)}=3$.
We can assume, without loss of generality,
that $||\alpha_1||^2=2$ and $||\alpha_2||^2=\frac23$.
With this choice, the Gram matrix for the fundamental weights $\wp_1,\wp_2$ is
$$
S
=
D
A^{-1}
=
\left(\begin{array}{cc}
2&1 \\
1&\frac23
\end{array}\right)
\quad
\text{where}
\quad
D
=
\left(\begin{array}{cc}
1&0 \\
0&\frac13
\end{array}\right)
\,.
$$

By a direct computation, 
the maximal root is
$$
\theta
=
2\alpha_1+3\alpha_2=\wp_1
\,,
$$
and the sum of positive  roots is
\begin{equation}\label{deltaG2}
2\delta
=
6\alpha_1+10\alpha_2
\,.
\end{equation}

Since $A^{-1}$ has integer positive entries, $\Lambda^+\subset P^+$
so that
$P^+\cap\Lambda^+=\bb Z_+\wp_1+\bb Z_+\wp_2$.
For $\mu=n_1\wp_1+n_2\wp_2$ we have
$$
(\mu|\mu+2\delta)
=
n^T S n+R^TDn
\,,
$$
where $n^T=(n_1,n_2)$ and $R^T=(6,10)$ (cf. \eqref{deltaG2}).

Since both the matrix $S$ and the vector $R$ have positive entries,
a direct comparison shows that the minimum of $(\mu|\mu+2\delta)$
is achieved when $n_i=\delta_{2,i}$,
i.e. when $\mu=\wp_2$.
The corresponding value is
$(\wp_2|\wp_2+2\delta)
=
4$.
On the other hand, we have 
$(\theta|\theta+2\delta)
=
8$.
As a consequence $g_0=\frac12$, as claimed in in Theorem \ref{t:mrl}..

We conclude
by exhibiting a representation of $\mf g$ for which 
$\gap(-L)=\frac12$ and another representation for which $\gap(-L)=1$.
For the 
$7$-dimensional representation $V=V_{\wp_2}$ we have
$\End V\simeq V_0\oplus V_{\wp_1}\oplus V_{\wp_2}\oplus V_{2\wp_1}$,
so that, in this case, $V_{\wp_2}$ does appear in the decomposition 
and $\gap(-L)=\frac12$.
On the other hand, for the adjoint representation $V=V_{\wp_1}$
we have
$\End V\simeq V_0\oplus V_{\wp_1}\oplus V_{2\wp_1}\oplus V_{2\wp_2}\oplus V_{3\wp_2}$.
Since, in this case, $V_{\wp_2}$ does not appear,
we have $\gap(-L)=1$.

\section*{Acknowledgments}
Our interest in quantum Markov semigroups has been stimulated
by discussions with A. Faggionato, D. Gabrielli, G. Jona-Lasinio and C. Presilla.
We also thank an anonymous referee for several helpful comments and for pointing out the connection between the QMS analyzed in Section~\ref{s:la} and those introduced in \cite{Ba}.


\end{document}